\documentclass[a4paper,reqno]{amsart}
\usepackage{amssymb} 
\usepackage{hyperref}
\usepackage{eucal}
\usepackage{graphicx}
\usepackage{epsfig}
\numberwithin{equation}{section}
\newtheorem{definition}{Definition}[section]
\newtheorem{theo}[definition]{Theorem}
\newtheorem{prop}[definition]{Proposition}
\newtheorem{remark}[definition]{Remark}

  \vsize
27.9truecm \hsize 21.59truecm \hoffset -0.4truecm
\newcommand{\lcf}{\lbrack\! \lbrack}
\newcommand{\rcf}{\rbrack\! \rbrack}\newcommand{\TT}{\mathcal{T}}
\newcommand{\GG}{\mathcal{G}}

\newcommand{\HH}{\mathcal{H}}
\newcommand{\dd}{\mathfrak{d}}
\def\a{\alpha}
\newcommand{\R}{\mathbb{R}}
\begin{document}

\title{The ubiquity of the symplectic hamiltonian equations in mechanics}

\author[P.\ Balseiro]{Paula Balseiro}
\address{Paula Balseiro:
Instituto de Ciencias Matem\'aticas (CSIC-UAM-UC3M-UCM),  Serrano 123, 28006
Madrid, Spain and CONICET Argentina.} \email{poi@mate.unlp.edu.ar}

\author[M.\ de Le\'on]{Manuel de Le\'on}
\address{Manuel de Le\'on:
ULL-CSIC Geometr\'{\i}a Diferencial y Mec\'anica Geom\'etrica\\
Instituto de Ciencias Matem\'aticas (CSIC-UAM-UC3M-UCM),  Serrano 123, 28006
Madrid, Spain} \email{mdeleon@imaff.cfmac.csic.es}

\author[J.\ C.\ Marrero]{Juan C.\ Marrero}
\address{Juan C.\ Marrero:
ULL-CSIC Geometr\'{\i}a Diferencial y Mec\'anica Geom\'etrica\\
Departamento de Matem\'atica Fundamental, Facultad de
Ma\-te\-m\'a\-ti\-cas, Universidad de la Laguna, La Laguna,
Tenerife, Canary Islands, Spain} \email{jcmarrer@ull.es}

\author[D.\ Mart\'{\i}n de Diego]{David Mart\'{\i}n de Diego}
\address{D.\ Mart\'{\i}n de Diego:
ULL-CSIC Geometr\'{\i}a Diferencial y Mec\'anica Geom\'etrica\\
Instituto de Ciencias Matem\'aticas (CSIC-UAM-UC3M-UCM),  Serrano 123, 28006
Madrid, Spain} \email{d.martin@imaff.cfmac.csic.es}

\thanks{This work has been partially supported by MEC (Spain)
Grants MTM 2006-03322, MTM 2007-62478, project ``Ingenio
Mathematica" (i-MATH) No. CSD 2006-00032 (Consolider-Ingenio 2010)
and S-0505/ESP/0158 of the CAM. The authors wish to thank David Iglesias for  helpful comments}

\begin{abstract}
In this paper,  we derive a ``hamiltonian formalism" for a
wide class of mechanical systems, including classical hamiltonian
systems, nonholonomic systems, some classes of servomechanism...  This
construction strongly relies in the geometry characterizing the
different systems. In particular, we obtain that the class of the
so-called algebroids covers a great variety of
mechanical  systems. Finally, as the main result, a hamiltonian symplectic
realization of  systems defined on
algebroids is obtained.
\end{abstract}


\maketitle


\baselineskip 6 mm

\section{Introduction}

One of the most important equations in Mathematics and Physics are
certainly the Hamilton equations:
\[
\dot{q}^i=\frac{\partial H}{\partial p_i},\qquad
\dot{p}_i=-\frac{\partial H}{\partial q^i},
\]
where $(q^i, p_i)$ are canonical coordinates. Symplectic geometry
allows us to write these equations in an intrinsic way (see
\cite{AM,LeRo}),
\begin{equation}\label{hamilton}
i_{X_H}\omega_Q=dH\; .
\end{equation}
In equation (\ref{hamilton}), $H: T^*Q\to \R$ represents a
Hamiltonian function defined on the cotangent bundle $T^*Q$  of a configuration
manifold $Q$ and $\omega_Q$ is the canonical symplectic form of the
cotangent bundle (in canonical coordinates, $\omega_Q=dq^i\wedge
dp_i$). The skew-symmetry of the canonical symplectic form leads
to conservative properties for the Hamiltonian vector field $X_H$
(preservation of the energy). On the contrary, in other type of
systems this conservative behavior is not required. For instance,
from the symmetry of a riemmanian metric it follows  dissipative
properties for the gradient vector field (see \cite{CoVaCr}).

It is an universal belief that  Equation (\ref{hamilton}) is only
valid for free Hamiltonian systems. For other type of systems,
Equation (\ref{hamilton}) is, in general, not longer valid (for
instance, in the presence of nonholonomic constraints or
dissipative forces, or in the case of gradient systems). In these
cases,  it is necessary to modify Equation (\ref{hamilton}) adding
some extra-terms of different nature: dissipative forces,
constraint forces, etc... (see \cite{BaSo,c,CeGr,Cort,CoVaCr,m}).

Our approximation adopts a new point of view. First,  it is
necessary to understand the underlying geometry  of Equation
(\ref{hamilton}) which will permit us to conclude that Hamilton's
equations have an ubiquity  property: \emph{many different
mechanical systems can be described by a  symplectic equation
constructed in the same way than in the standard case}. Of course,
this construction relies on the different geometry behind each
particular problem. We will show in this paper the main lines of
that construction. We should remark that,  for the particular
construction of Equation (\ref{hamilton}) in each different
mechanical problem, it will be necessary to introduce some
sophisticated geometric techniques like algebroids, prolongation
of structures, lifting of connections
\cite{CaLa,GG,GrUr2,GrabUrba97,LMM} among others,
 which are based on its underlying geometry.

Let us describe our method with more details. If we extract the
geometric elements that appear in Equation  (\ref{hamilton}) we
observe that $\omega_Q=dq^i\wedge dp_i$ is derived from the
Liouville 1-form $\lambda_Q=p_idq^i$, more precisely, $\omega_Q=-d\lambda_Q$
(see \cite{AM,LeRo} for details). In symplectic geometry terms we
say that $(T^*Q, \omega_Q)$ is an exact symplectic manifold. This
structure induces a linear Poisson tensor field $\Pi_{T^*Q}$ on
$T^*Q$ defined by
\[
\Pi_{T^*Q}(dF, dG)=\omega_Q(X_F, X_G)\; ,
\]
where $X_F$ and $X_G$ are the hamiltonian vector fields
corresponding to the functions $F: T^*Q\to \R$ and $G: T^*Q\to
\R$, respectively. In canonical coordinates
\[
\Pi_{T^*Q}=\frac{\partial}{\partial q^i}\wedge
\frac{\partial}{\partial p_i}.
\]

A trivial, but interesting, comment is that the classical bracket
of vector fields (the standard Lie bracket) is induced by the
linear Poisson tensor $\Pi_{T^*Q}$ (and viceversa). In fact, there
exists a one-to-one correspondence between the space of vector
fields on $Q$ and the space of linear functions on $T^*Q$. Indeed,
for each vector field $X\in {\mathfrak X}(Q)$ the corresponding
linear function $\widehat{X}: T^*Q\to \R$ is given by
\[
\widehat{X}(\kappa_q)=\langle \kappa_q, X(q)\rangle,
\]
where $\langle\; ,\rangle$ is the natural pairing between vectors and
covectors, and $\kappa_q\in T^*_qQ$. Therefore, for $X, Y\in
{\mathfrak X}(Q)$, the bracket of the two vector fields $X$ and
$Y$ is characterized as the unique vector field associated to the
linear function $-\Pi_{T^*Q}(d\hat{X}, d\hat{Y})$. Observe that in
coordinates
\[
\Pi_{T^*Q}(d\hat{X}, d\hat{Y})=p_j\left(\frac{ \partial
X^j}{\partial q^i}Y^i- \frac{ \partial Y^j}{\partial
q^i}X^i\right)=-\widehat{[X, Y]}\; .
\]
In a schematic way we have
$$  \text{linear Poisson tensor $\Pi_{T^*Q}$}
\quad \longleftrightarrow\quad \text{standard Lie bracket on $Q$}.
$$
That is, we have that there exists a one-to-one correspondence
between linear Poisson tensors on $T^*Q$ and  Lie algebra structures of vector fields on $Q$ (see \cite{Cou}).

As a preliminary conclusion,  classical hamiltonian
formulations strongly relies on the standard Lie bracket of vector
fields (or equivalently, the linear Poisson tensor on $T^*Q$). Modifications of this bracket (or the associated linear tensor) will presumably change
the properties of the dynamics. For example, if we do not impose
the skew-symmetry of the bracket we will have a dissipative
behavior, since, in the cotangent bundle, we will obtain a
2-contravariant linear tensor field which is not necessarily
skew-symmetric \cite{CoVaCr}. Another property that it is possible
to drop is the Jacobi identity, which is related with the
preservation of the symplectic form by the flow of the hamiltonian
vector field. In many interesting cases, as for instance
nonholonomic mechanics (see \cite{CdLMM2007,CorMar,GG}), it is
well known that the Jacobi identity is equivalent to the
integrability of the constraints, that is, to holonomic mechanics.
Since our objective is to obtain a geometric framework including
all these cases, it is necessary to work without imposing Jacobi
identity, from the beginning,  to our tensor field or associated
bracket. Moreover, the role of the tangent bundle is not
essential, and we may change it for an arbitrary vector bundle,
and then the linear contravariant tensor field will be now defined
on its dual bundle $E^*$ (see \cite{Weinstein99}).

We will show that the category of  algebroids is
general enough to cover all the cases that we want to analyze. An
 algebroid (see \cite{GGU,GrUr2,GrabUrba97,IvOp,OrPl}) is, roughly
speaking, a vector bundle $\tau_E: E\to Q$,  equipped with a
bilinear bracket of sections $B_E: \Gamma(\tau_E)\times
\Gamma(\tau_E)\to \Gamma(\tau_E)$ and two vector bundle morphisms
$\rho_E^l: E\to TQ$ and $\rho_E^r: E\to TQ$ satisfying a
Leibniz-type property (see (\ref{Eqfun})). Observe that properties
like skew-symmetry or Jacobi identity are not considered in this
category. This general structure is equivalent to give a linear
2-contravariant tensor field $\Pi_{E^*}$ on its dual bundle $E^*$.
In conclusion, we have that
$$  \text{linear contravariant tensor field $\Pi_{E^*}$} \quad
\longleftrightarrow\quad \text{ algebroid structure on $E$} \,.
$$

The main objective of this paper is to show that the general
construction of the hamiltonian symplectic formalism in classical
mechanics remains essentially unchanged starting from the more
general framework of algebroids. This result, which is proved in three
Theorems (Theorems \ref{TeoExact}, \ref{TeoHamiltonian} and
\ref{Teocloded}), constitutes  the core of our paper. Additionally,
we show how to apply these new techniques  to several examples of
interest: (generalized)-nonholonomic mechanics, dissipative
systems, and gradient systems.

The paper is structured as follows. In Section
\ref{subSecAlmosLeibniz} we define the notion of an algebroid and
relate this concept with Hamilton equations for general linear
2-contravariant tensor fields. Moreover, some examples of interest
are considered: gradient extension of dynamical systems,
nonholonomic mechanics and generalized nonholonomic mechanics. In
Section \ref{section-exact}, it is introduced the notion of an
exact symplectic algebroid, structure that will be necessary to
formulate the main result of the paper: the construction of a
symplectic formulation of hamiltonian mechanics in the context of
algebroids in Sections \ref{section-main} and  \ref{Section5}.
Finally, we apply the precedent results to the examples considered
in Section \ref{subSecAlmosLeibniz}.

\bigskip

\section{Algebroids and Hamiltonian Mechanics}\label{subSecAlmosLeibniz}
\setcounter{equation}{0}

Let $\tau_E : E \rightarrow Q$ be a real vector bundle over a
manifold $Q$ and $\Gamma(\tau_{E})$ be the space of sections of
$\tau_{E}: E \rightarrow Q$.

\begin{definition}
An {algebroid structure} on $E$ is a $\R$-bilinear bracket
$$B_E : \Gamma(\tau_E) \times \Gamma(\tau_E) \rightarrow
\Gamma(\tau_E)$$ together with two vector bundles morphisms
$\rho^l_E, \rho^r_E: E \rightarrow TQ$ ({left} and {right anchors})
such that
\begin{equation}
B_E (f \sigma, f'\sigma') = f \rho_E^l(\sigma)(f') \sigma' -  f'
\rho_E^r(\sigma')(f) \sigma + ff' B_E(\sigma, \sigma') \label{Eqfun}
\end{equation} for $f, f' \in C^{\infty}(Q)$ and $\sigma, \sigma' \in \Gamma(\tau_E)$.
\end{definition}

The algebroid structure on $\tau_E : E \rightarrow Q$ was defined in
\cite{GGU,GrUr2,GrabUrba97} and it is called a {\it Leibniz
algebroid}
 in \cite{IvOp,OrPl}.

If the $\R$-bilinear bracket $B_E$ is skew-symmetric we have a
\emph{skew-symmetric algebroid structure} \cite{GGU} (an
\emph{almost Lie algebroid structure} in the terminology of
\cite{LeMaMa} or an almost-Lie structure in the terminology of
\cite{Po}). In such a case, the left anchor coincides with the
right anchor: $\rho^l_E=\rho^r_E$. In the sequel, we will denote
the bracket of sections in this skew-symmetric case by $\lcf\; ,\;
\rcf_E$. On the other hand if the bracket $\lcf\; ,\; \rcf_E$
defines a Lie algebra structure on the space $\Gamma(\tau_E)$ then
the pair $(\lcf\; ,\; \rcf_E, \rho_E= \rho^l_E = \rho^r_E)$ is a
\emph{Lie algebroid structure} on the vector bundle $\tau_E: E
\rightarrow Q$ (see \cite{Mac}).

Another interesting case  is when the $\R$-bilinear bracket $B_E$
is symmetric, then  we have a \emph{symmetric algebroid
structure}. In such a case, $\rho^l_E=-\rho^r_E$.

Now, note that there exists a one-to-one correspondence between
the space $\Gamma(\tau_{E})$ of sections of the vector bundle
$\tau_{E}: E \to Q$ and the space of linear functions on $E^*$. In
fact, if $\sigma \in \Gamma(\tau_{E})$ then the corresponding
linear function $\widehat{\sigma}$ on $E^*$ is given by
\[
\widehat{\sigma}(\kappa) =
\kappa(\sigma(\tau_{E^*}(\kappa)))=\langle \kappa,
\sigma(\tau_{E^*}(\kappa))\rangle , \; \; \mbox{ for } \kappa \in
E^*,
\]
where $\tau_{E^*}: E^* \to Q$ is the vector bundle projection.

An algebroid structure $(B_E, \rho^l_E , \rho^r_E)$ on a vector
bundle $\tau_E: E \rightarrow Q$ induces a linear tensor $\Pi_{E^*}$
of type (2,0) on $E^*$. In fact, if $\{\cdot , \cdot
\}_{\Pi_{E^*}}: C^{\infty}(E^*) \times C^{\infty}(E^*) \rightarrow
C^{\infty}(E^*)$ is the induced bracket of functions given by
$$\{\varphi, \psi \}_{\Pi_{E^*}} = \Pi_{E^*} ( d \varphi, d \psi ),
\ \ \mbox{for} \ \varphi, \psi \in C^{\infty}(E^*),$$ then we have
that
$$\{\widehat{\sigma}, \widehat{\sigma'}
\}_{\Pi_{E^*}} = - \widehat{B_E( \sigma , \sigma' )},  \ \ \ \ \ \
\ \{\widehat{\sigma}, f' \circ \tau_{E^*} \}_{\Pi_{E^*}} = -
\rho_E^l(\sigma)(f') \circ \tau_{E^*}$$
\begin{equation} \{ f \circ
\tau_{E^*}, \widehat{\sigma'}\}_{\Pi_{E^*}} = \rho_E^r(\sigma')(f)
\circ \tau_{E^*}, \ \ \ \ \ \ \ \{ f \circ \tau_{E^*}, f' \circ
\tau_{E^*} \}_{\Pi_{E^*}} =0, \label{Pistar}
\end{equation}
for $\sigma, \sigma' \in \Gamma(\tau_E)$ and $f, f' \in
C^{\infty}(Q)$ (see \cite{GGU,GrUr2,GrabUrba97}).

A curve $\gamma: I \to E$ is $\rho^l_E$-\emph{admissible}
(respectively, $\rho^r_E$-\emph{admissible}) if $\displaystyle
\frac{d}{dt}(\tau_{E} \circ \gamma) = \rho^l_{E} \circ \gamma$
(respectively, $\displaystyle \frac{d}{dt}(\tau_{E} \circ \gamma) =
\rho^r_{E} \circ \gamma$).

In the particular case when $E$ is a skew-symmetric algebroid it
follows that $\Pi_{E^*}$ is a linear 2-vector on $E^*$ (or an
\emph{almost Poisson structure} on $E^*$ in the terminology of
\cite{LeMaMa}). If $E$ is a Lie algebroid, the bracket $\{ \cdot ,
\cdot \}_{\Pi_{E^*}}$ satisfies the Jacobi identity and
$\Pi_{E^*}$ is a Poisson structure on $E^*$ (see
\cite{Cou,LeMaMa,Li,Weinstein99}).

Now let $H : E^* \rightarrow \R$ be a \emph{Hamiltonian function}
on $E^*$. Then one may consider the \emph{Hamiltonian vector
field} $\HH_H^{\Pi_{E^*}}$ of $H$ with respect to $\Pi_{E^*}$,
that is,
$$\HH_H^{\Pi_{E^*}} (F) = - \{ H, F \}_{\Pi_{E^*}} , \ \
\mbox{for} \ F \in C^{\infty}(E^*).$$

The integral curves of the vector field $\HH_H^{\Pi_{E^*}}$ are
the solutions of the \emph{Hamilton equations} for $H$. Next, we
will obtain some local expressions.

Suppose that $(q^i)$ are local coordinates on $Q$ and that
$\{\sigma_{\a}\}$ is a local basis of the space $\Gamma(\tau_E)$
such that \begin{equation} B_E(\sigma_{\a}, \sigma_{\beta}) =
{(B_E)}_{\a \beta}^{\gamma} e_{\gamma}, \ \ \ \rho_E^l (\sigma_{\a})
= {(\rho_E^l)}_{\a}^i \frac{\partial}{\partial q^i}, \ \ \ \rho_E^r
(\sigma_{\beta}) = {(\rho_E^r)}_{\beta}^j \frac{\partial}{\partial
q^j} .\label{Lostrfu} \end{equation}

The local functions ${(B_E)}_{\a \beta}^{\gamma},
{(\rho_E^l)}_{\a}^i$ and ${(\rho_E^r)}_{\beta}^j$ are called the
\emph{local structure functions} of algebroid $\tau_E: E \rightarrow
Q$.

Denote by $(q^i, p_{\a})$ the induced local coordinates on $E^*$.
Then using (\ref{Pistar}) and (\ref{Lostrfu}), it follows that
$$\Pi_{E^*} = {(\rho_E^r)}_{\a}^i \frac{\partial}{\partial q^i}
\otimes \frac{\partial}{\partial p_{\a}} - {(\rho_E^l)}_{\beta}^j
\frac{\partial}{\partial p_{\beta}} \otimes \frac{\partial}{\partial
q^j} - {(B_E)}_{\a \beta}^{\gamma} p_{\gamma}
\frac{\partial}{\partial p_{\a}} \otimes \frac{\partial}{\partial
p_{\beta}}.$$

Therefore, the Hamiltonian vector field of $H$ is given by
\begin{equation} \HH_H^{\Pi_{E^*}} = {(\rho_E^l)}_{\a}^i
\frac{\partial H}{\partial p_{\a}}\frac{\partial}{\partial q^i}  -
\left( {(\rho_E^r)}_{\beta}^j \frac{\partial H}{\partial q^j} -
{(B_E)}_{\a \beta}^{\gamma} p_{\gamma} \frac{\partial H}{\partial
p_{\a}}  \right) \frac{\partial}{\partial p_{\beta}}, \label{LocHvf}
\end{equation} which implies that the local expression of the
Hamilton equations is $$\frac{dq^i}{dt} = {(\rho_E^l)}_{\a}^i
\frac{\partial H}{\partial p_{\a}}, \ \ \ \ \ \frac{dp_{\beta}}{dt}
= - \left\{ {(\rho_E^r)}_{\beta}^j \frac{\partial H}{\partial q^j} -
{(B_E)}_{\a \beta}^{\gamma} p_{\gamma} \frac{\partial H}{\partial
p_{\a}}  \right\}.$$

\begin{remark}
{\rm Working with not in general skew-symmetric or symmetric
tensors it is possible to distinguish two different dynamics (one
on the left and one on the right). Therefore, we can obtain also a
different Hamiltonian vector field $\widetilde{\HH}_H^{\Pi_{E^*}}$
of $H$ with respect to $\Pi_{E^*}$:
$$\widetilde{\HH}_H^{\Pi_{E^*}} (F) =\{ F, H \}_{\Pi_{E^*}} , \ \ \mbox{for} \ F \in C^{\infty}(E^*).$$
In coordinates, \[ \widetilde{\HH}_H^{\Pi_{E^*}} =
{(\rho_E^r)}_{\a}^i \frac{\partial H}{\partial
p_{\a}}\frac{\partial}{\partial q^i}  - \left(
{(\rho_E^l)}_{\beta}^j \frac{\partial H}{\partial q^j} +
{(B_E)}_{\beta\a}^{\gamma} p_{\gamma} \frac{\partial H}{\partial
p_{\a}}  \right) \frac{\partial}{\partial p_{\beta}}.
\]
In the sequel we only consider the vector field
$\HH_H^{\Pi_{E^*}}$ since the analysis for
$\widetilde{\HH}_H^{\Pi_{E^*}}$ is similar.

 }
 \end{remark}

\subsection{First example. A symmetric case: Gradient extension of a dynamical system}\label{grad-ext}

Let $Q$ be an $n$-dimensional manifold. Let ${\mathcal G}$ be a
riemannian metric on $Q$, i.e, a positive-definite  symmetric
$(0,2)$-tensor on $Q$. Associated to ${\mathcal G}$ we have the
associated musical isomorphisms
\begin{eqnarray*}
&&\flat_{\mathcal G}: {\mathfrak X}(Q)\longrightarrow \Lambda^1(Q), \qquad \flat_{\mathcal G}(X)(Y)={\mathcal G}(X,Y),\\
&&\sharp_{\mathcal G}(\mu)=\flat_{\mathcal G}^{-1}(\mu)
\end{eqnarray*}
where $X, Y\in {\mathfrak X}(Q)$ and $\mu\in \Lambda^1(Q)$. In
coordinates $(q^i)$ on $Q$,  the metric is expressed as ${\mathcal
G}={\mathcal G}_{ij}(q)dq^i\otimes dq^j$.

Fixed a function $f\in C^{\infty}(M)$, it is defined the gradient
vector field associated to $f$ as $\hbox{grad}_{\mathcal
G}(f)=\sharp_{\mathcal G}(df)$. In coordinates,
\[
\hbox{grad}_{\mathcal G}(f)={\mathcal G}^{ij}\frac{\partial
f}{\partial q^j}\frac{\partial}{\partial q^i}
\]
where $({\mathcal G}^{ij})$ is the inverse matrix of $({\mathcal G}_{ij})$.

Associated with the metric ${\mathcal G}$ there is an affine
connection $\nabla^{\mathcal G}$, called the \emph{Levi-Civita
connection} determined by:
\begin{equation}\label{levi}
\begin{array}{l}
 [X,Y]= \nabla^{\mathcal G}_X Y-\nabla^{\mathcal G}_Y X \hbox{   (symmetry)}\\
 X( {\mathcal G}(Y,Z))={\mathcal G}(\nabla^{\mathcal G}_X Y, Z)+{\mathcal G}(Y, \nabla^{\mathcal G}_X Z)\hbox{   (metricity)}\; ,
 \end{array}
 \end{equation}
where $X, Y, Z\in {\mathfrak X}(Q)$. Locally, $ \nabla^{\mathcal
G}_{\frac{\partial}{\partial q^j}}\frac{\partial}{\partial
q^k}=\Gamma_{jk}^i\frac{\partial}{\partial q^i} $ where the
Christoffel symbols $\Gamma_{jk}^i$ of $\nabla^{\mathcal G}$ are given by
\[
\Gamma_{jk}^i=\frac{1}{2}{\mathcal G}^{ik'}\left( \frac{\partial
{\mathcal G}_{k'j}}{\partial q^k}+ \frac{\partial {\mathcal
G}_{k'k}}{\partial q^j}-\frac{\partial {\mathcal G}_{jk}}{\partial
q^{k'}}\right)\; .
\]

Consider now the \emph{symmetric product}:
\[
B_{TQ}(X,Y)=\nabla^{\mathcal G}_X Y +\nabla^{\mathcal G}_Y X\; \quad
X, Y\in {\mathfrak X}(Q).
\]
Locally,
\[
B_{TQ}(\frac{\partial}{\partial q^j},\frac{\partial}{\partial
q^k})=\left(\Gamma_{jk}^i+\Gamma_{kj}^i\right)\frac{\partial}{\partial
q^i}=2\Gamma_{jk}^i\frac{\partial}{\partial q^i}.
\]
It is well known that the symmetric product is an element crucial in
the study of various aspects  of mechanical control systems such us
controllability, motion planning, (see for example \cite{BulLewis}) and also characterize when a
distribution is geodesically invariant \cite{Lewis}. Now define the left and
right anchors by, $\rho_{TQ}^l=\hbox{id}_{TQ}$ and
$\rho_{TQ}^r=-\hbox{id}_{TQ}$. The tangent bundle equipped with
$(B_{TQ}, \rho^l_{TQ} , \rho^r_{TQ})$ is a  (symmetric) algebroid.

This structure induces  a linear tensor $\Pi_{T^*Q}$ of type (2,0)
on $T^*Q$. In local coordinates $(q^i, p_i)$ on $T^*Q$, the
bracket relations induced by this tensor field are:
\begin{eqnarray*}
 &\{q^i\, , q^j\}_{\Pi_{T^*Q}}=0,\qquad &\{q^i\, , p_j\}_{\Pi_{T^*Q}}=-\delta_{j}^{i}\\
 &\{p_i\, , q^j\}_{\Pi_{T^*Q}}=-\delta_{i}^{j},\qquad &\{p_i\, , p_j\}_{\Pi_{T^*Q}}=-2p_k\Gamma_{ij}^k\; .
 \end{eqnarray*}
Given an arbitrary vector field $X\in {\mathfrak X}(Q)$ one may
define the function $H_X:T^*Q\longrightarrow \R$ by
$H_X(\kappa)=\langle \kappa, X_q\rangle$, for $\kappa \in T^*_qQ$,
that is, $H_X=\widehat{X}$. In coordinates, $H_X(q^i,
p_i)=p_iX^i(q)$.

The hamiltonian vector field $\HH_{H_X}^{\Pi_{T^*Q}}$ is
 \begin{equation*} \HH_{H_X}^{\Pi_{T^*Q}} = X^i(q)\frac{\partial }{\partial q^i}  +
p_k\left( \frac{\partial X^k}{\partial q^j} + 2\Gamma_{ij}^k X^i
\right) \frac{\partial}{\partial p_{j}}, \label{LocHvf-1}
\end{equation*}
 The equations for its integral curves  are:
 \begin{eqnarray}
 \dot{q}^i&=&X^i(q)\nonumber \\
 \dot{p}_j&=&p_k\left( \frac{\partial X^k}{\partial q^j} + 2\Gamma_{ij}^k X^i \right)\;
 . \label{HamilGrad}
 \end{eqnarray}
These equations are the \emph{gradient extension} of the nonlinear
equation $\dot{q}^i=X^i(q)$ (see \cite{CoVaCr}).

\subsection{Second example. An antisymmetric case: Nonholonomic Mechanics}\label{qaz}

Let $\tau_{E}:E \to Q$ be a Lie algebroid over a manifold $Q$ and
denote by $(\lcf\cdot ,\cdot \rcf_E, \rho_{E})$ the Lie algebroid
structure on $E$.

Following \cite{LeMaMa}, a \emph{mechanical system subjected to
linear nonholonomic constraints on $E$} is a pair $(L, D)$, where
$L: E \to \R$ is a \emph{Lagrangian function of mechanical type},
that is,
\begin{equation}\label{lagrangian}
L(e) = \displaystyle \frac{1}{2} {\mathcal G}(e, e) -
V(\tau_{E}(e)), \; \; \; \mbox{ for } e\in E,
\end{equation}
with  ${\mathcal G}: E \times_{Q} E \to \R$ a bundle metric on $E$
and $D$ the total space of a vector subbundle $\tau_{D}: D \to Q$
of $E$ such that $rank D = m$. The vector subbundle $D$ is said to
be the \emph{ constraint subbundle}.

Denote by $i_{D}: D \to E$ the canonical inclusion and consider
the orthogonal decomposition $E = D \oplus D^{\perp}$ and the
associated orthogonal projectors $P: E \to D$ and $Q: E \to
D^{\perp}$.

The Levita-Civita connection $\nabla^{\mathcal G}:
\Gamma(\tau_E)\times \Gamma(\tau_E)\to \Gamma(\tau_E)$ associated
to the bundle metric ${\mathcal G}$ is defined in a similar way
than in (\ref{levi}) (see \cite{CorMar}). It is determined by the
formula:
\[
\begin{array}{rcl}
2 {\mathcal G}(\nabla_{\sigma}^{\mathcal G}\sigma', \sigma'') & = &
\rho_{E}(\sigma)({\mathcal G}(\sigma', \sigma'')) + \rho_{E}(\sigma')({\mathcal G}(\sigma, \sigma''))
- \rho_{E}(\sigma'')({\mathcal G}(\sigma, \sigma')) \\
&& +  {\mathcal G}(\sigma, \lcf \sigma'', \sigma'\rcf_{E}) + {\mathcal G}(\sigma',\lcf \sigma'',
\sigma\rcf_{E}) - {\mathcal G}(\sigma'', \lcf \sigma', \sigma\rcf_{E})
\end{array}
\]
for $\sigma, \sigma', \sigma'' \in \Gamma(\tau_E)$. The solutions
of the nonholonomic problem are the $\rho_E$-admissible curves
$\gamma: I\longrightarrow D$ such that (see \cite{CdLMM2007,CorMar})
\begin{equation}\label{noholonoma}
\nabla^{\mathcal G}_{\gamma(t)}(\gamma(t)) + \hbox{grad}_{\mathcal
G}V(\tau_D(\gamma(t)))\in D^{\perp}_{\tau_D(\gamma(t))}\; .
\end{equation}
Here, $\hbox{grad}_{\mathcal G}V$ is the section of $\tau_{E}: E
\rightarrow Q$ which is characterized by
\[
{\mathcal G}(\hbox{grad}_{\mathcal G}V, \sigma) =
\rho_{E}(\sigma)(V), \; \; \mbox{ for } \sigma \in
\Gamma(\tau_{E}).
\]

Now, we will derive the equations of motion (\ref{noholonoma}) using
the general procedure introduced in Section
\ref{subSecAlmosLeibniz}. First, we define on the vector bundle
$\tau_D: D\longrightarrow Q$ the following \emph{skew-symmetric
algebroid structure}:
\begin{equation}\label{alLieD}
\lcf\sigma, \sigma'\rcf_D= P(\lcf i_{D}\circ \sigma, i_{D}\circ \sigma')\rcf_E), \;
\; \; \rho_{D}(\sigma) = \rho_{E}(i_{D}\circ \sigma),
\end{equation}
for $\sigma, \sigma' \in \Gamma(\tau_{D})$.

This skew-symmetric algebroid induces a linear
almost-Poisson tensor field $\Pi_{D^*}$ on the dual bundle $D^*$. In
\cite{LeMaMa}, it is shown that this structure is also induced from
the linear Poisson bracket $\{\cdot , \cdot \}_{\Pi_{E^*}}$ on
$E^*$:
\begin{equation}\label{laPbDstar}
\{\varphi, \psi\}_{\Pi_{D^*}} = \{\varphi \circ i_{D}^*,\psi \circ
i_{D}^*\}_{\Pi_{E^*}} \circ P^*,
\end{equation}
for $\varphi, \psi \in C^{\infty}(D^*)$, where $i_{D}^*: E^* \to
D^*$ and $P^*: D^* \to E^*$ are the dual maps of the monomorphism
$i_{D}: D \to E$ and the projector $P: E \to D$, respectively.

Next, suppose that $(q^i)$ are local coordinates on an open subset
$U$ of $Q$ and that $\{\sigma_{\a}\} = \{\sigma_{a}, \sigma_{A}\}$
is a basis of sections of the vector bundle $\tau_{E}^{-1}(U) \to
U$ such that $\{\sigma_{a}\}$ (respectively, $\{\sigma_{A}\}$) is
an orthonormal basis of sections of the vector subbundle
$\tau_{D}^{-1}(U) \to U$ (respectively, $\tau_{D^{\perp}}^{-1}(U)
\to U$). We will denote by $(q^i, v^{\a}) = (q^i, v^{a}, v^A)$ the
corresponding local coordinates on $E$. Observe, now, that in this
coordinates the equations defining $D$ as a vector subbundle of
$E$ are:
\[
v^A = 0, \; \; \; \mbox{ for all } A.
\]

On the other hand, $\lcf\sigma_b,\sigma_c \rcf_D=C_{bc}^a\sigma_a$
and $\rho_D(\sigma_a) = (\rho_D)^i_a\frac{\partial}{\partial q^i} =
(\rho_E)^i_a\frac{\partial}{\partial q^i}$, where
$C_{\alpha\beta}^\gamma$ and $(\rho_E)^i_\a$ are the structure
functions of the Lie algebroid $\tau_E: E\longrightarrow Q$ with
respect to the local basis $\{\sigma_{\alpha}\}$. If we denote by
$(q^i, p_{a})$ the induced local coordinates on $D^*$, then $\{\; ,
\; \}_{\Pi_{D^*}}$ is determined by the following bracket relations:
\[
\{q^i\, , q^j\}_{\Pi_{D^*}}=0,\qquad \{q^i\, ,
p_a\}_{\Pi_{D^*}}=-\{p_a, q^i\}_{\Pi_{D^*}} = (\rho_D)^i_a, \qquad
 \{p_a\, , p_b\}_{\Pi_{D^*}}=-p_c C^c_{ab}\; .
 \]

 Now, denote by $\langle\; ,\rangle_{\mathcal G}: E^* \times_{Q} E^* \to \R$ the bundle metric
 on $E^*$ induced by ${\mathcal G}$.   Then, define the hamiltonian function $H_{E^*}:
 E^*\longrightarrow  \R$ (\emph{the Hamiltonian energy}) by:
 \[
 H_{E^*}(\kappa_q)=\frac{1}{2}\langle\kappa_q ,\kappa_q\rangle_{\mathcal G}+ V(q),\qquad \kappa_q\in E^*_q
 \]
 Next, we consider \emph{the constrained Hamiltonian function}
 $H_{D^*}$ on $D^*$ given by
 \begin{equation}\label{zxc}
 H_{D^*}=H_{E^*}\circ P^*: D^*\longrightarrow \R.
 \end{equation}

 Therefore,
  \begin{equation}\label{HhLambda-1}
\HH_{H_{D^*}}^{\Pi_{D^*}} = (\rho_D)^i_{a} \displaystyle
\frac{\partial H_{D^*}}{\partial p_{a}} \frac{\partial}{\partial
q^{i}} - ((\rho_D)^{i}_{b} \frac{\partial H_{D^*}}{\partial q^i} -
C_{ab}^{c}p_{c} \frac{\partial H_{D^*}}{\partial p_{a}})
\frac{\partial}{\partial p_{b}},
\end{equation}
and the  Hamilton equations are
\[
\displaystyle \frac{dq^i}{dt} = (\rho_D)^i_{a} \frac{\partial
H_{D^*}}{\partial p_{a}}, \; \; \; \frac{dp_{b}}{dt} = -
((\rho_D)^{i}_{b} \frac{\partial H_{D^*}}{\partial q^i} -
C_{ab}^{c}p_{c} \frac{\partial H_{D^*}}{\partial
p_{a}}).
\]

 In the induced local coordinates:
\[
H_{D^*}(q^i, p_a)=\frac{1}{2}\sum_{a=1}^m p_a^2+V(q^i),
\]
and the Hamilton equations are
 \begin{equation}\label{13'}
\displaystyle \frac{dq^i}{dt} = (\rho_D)^i_{a} p_a, \; \; \;
\frac{dp_{b}}{dt} = - \left((\rho_D)^{i}_{b} \frac{\partial
V}{\partial q^i} - C_{ab}^{c}p_{c} p_a\right).
\end{equation}

\begin{remark}{\rm
\emph{The Legendre transformation} associated with the Lagrangian
function $L$ is just the musical isomorphism $\flat_{\mathcal G}:
E \rightarrow E^*$ induced by the bundle metric ${\mathcal G}$
(for the definition of the Legendre transformation associated with
a Lagrangian function $L: E \rightarrow \mathbb{R}$, see
\cite{LMM}). \emph{The constrained Legendre transformation}
$\mbox{Leg}_{(L, D)}: D \rightarrow D^*$ associated with the
nonholonomic system $(L, D)$ is the vector bundle isomorphism
induced by the restriction of ${\mathcal G}$ to the vector
subbundle $\tau_{D}: D \rightarrow Q$ (see \cite{LeMaMa}). In
other words,
\[
\mbox{Leg}_{(L, D)} = i_{D}^* \circ \mbox{Leg}_{L} \circ i_{D} =
i_{D}^* \circ \flat_{\mathcal G} \circ i_{D}.
\]
Thus, it is easy to prove that
\[
\mbox{Leg}_{L}(q^i, v^{\alpha}) = (q^i, v^{\alpha}), \; \; \;
\mbox{Leg}_{(L, D)}(q^i, v^a) = (q^i, v^a).
\]
On the other hand, from (\ref{noholonoma}), it follows that a
curve
\[
\gamma: I \rightarrow D, \; \; \; \; \gamma(t) = (q^{i}(t),
v^{a}(t))
\]
is a solution of the nonholonomic problem if and only if
\[
\displaystyle \frac{dq^i}{dt} = (\rho_D)^{i}_{a}v^a, \; \; \;
\frac{dv^b}{dt} = - \left((\rho_D)^i_b \frac{\partial V}{\partial q^i} -
C_{ab}^c v^c v^a \right).
\]
Therefore, we deduce that a curve $\gamma: I \rightarrow D$ is a
solution of the nonholonomic problem if and only if the curve
$\mbox{Leg}_{(L, D)} \circ \gamma: I \rightarrow D^*$ is a
solution of the Hamilton equations (\ref{13'}).

Consequently, Eqs. (\ref{13'}) may be considered as the
\emph{nonholonomic Hamilton equations} for the constrained system
$(L, D)$.}
\end{remark}

\subsection{Third Example. Mixed cases:}
\ \

\noindent 2.3.1. {\bf Generalized nonholonomic systems on Lie
algebroids.}

\medskip

 In this section we will discuss Lagrangian systems on a Lie
algebroid $\tau_E: E \rightarrow Q$ subjected to generalized
nonholonomic constraints (see \cite{BaSo,c}).

As in the classical nonholonomic case, the \emph{kinematic
constraints} are described by a vector subbundle $\tau_D : D
\rightarrow Q$ of $\tau_E : E \rightarrow Q$. Therefore, to
determine the dynamics it is only necessary to fix a bundle of
reaction forces which vanishes on a vector subbundle
   $\tilde D$ of $E$. We have that, in general, $D \not=\tilde{D}$.
The case $D=\tilde{D}$, classical nonholonomic mechanics, was
studied in Subsection \ref{qaz}. The case $D\not= \tilde{D}$
appears in many interesting problems, for instance when the
restriction is realized by the action of a servo mechanism
\cite{m} or for rolling tyres as in \cite{c}.

We will call $\tau_ {\tilde D} : \tilde D \rightarrow Q$ the
\emph{variational subbundle}. It is important to note that $\tilde
D$ can not be deduced from the kinematic constraints, as the
virtual displacements are in classical nonholonomic mechanics.

Now, let ${\mathcal G}: E \times_{Q} E \rightarrow \mathbb{R}$ be
a bundle metric on the vector bundle $\tau_{E}: E \rightarrow Q$
and $V: Q \rightarrow \mathbb{R} \in C^{\infty}(Q)$. If
$\hbox{rank} D = \hbox{rank} \tilde{D}$ and $L: E \rightarrow
\mathbb{R}$ is the Lagrangian function of mechanical type given by
\[
L(e) = \displaystyle \frac{1}{2} {\mathcal G}(e, e) -
V(\tau_{E}(e)), \; \; \; \mbox{ for } e \in E,
\]
then the triple $(L, D, \tilde{D})$ will be called a
\emph{mechanical system subject to generalized linear nonholonomic
constraints} on $E$.

We will assume that $\tilde D$ satisfies the \emph{compatibility
condition}
\begin{equation}
E_q = D_q \oplus  \tilde{D}_q^{\perp}, \ \ \ \ \ \forall \, q \in Q,
\label{compCond}
\end{equation}
with $\tilde D_q^{\perp}$ the orthogonal complement of ${\tilde
D}_{q}$ in $E_{q}$ with respect to scalar product ${\mathcal
G}_{q}$. It is obvious that, in particular, this property holds in
the classical non-holonomic case $\tilde D = D$. In the general
case, we have that the equations of motion of such a system are
given by $\delta L_{\gamma(t)} \in \tilde
D^0_{\tau_D(\gamma(t))}$, for a $\rho_E$-admissible curve $\gamma
: I \rightarrow D$, or equivalently,
\begin{equation} \left\{ \begin{array}{l}
\displaystyle \frac{d}{dt}(\tau_{E} \circ \gamma) = \rho_{E}^{l}
\circ \gamma, \\ \vspace{-0.3cm} \\
 {\nabla}^{\mathcal G}_{\gamma(t)} \gamma(t)
+ grad_{\mathcal G}V(q(t)) \in \tilde D^{\perp}_{q(t)}, \\ \vspace{-0.3cm} \\
\gamma(t) \in D_{q(t)},
\end{array} \right. \label{gnhmotion}
\end{equation}
where $q = \tau_{D} \circ \gamma$.

Now, suppose that $(q^i)$ are local coordinates on an open subset
$U$ of $Q$ and that $\{\sigma_{\a}\} = \{\sigma_{a}, \sigma_{A}\}$
is a basis of sections of the vector bundle $\tau_{E}^{-1}(U) \to
U$ such that $\{\sigma_{a}\}$ (respectively, $\{\sigma_{A}\}$) is
an orthonormal basis of sections of the vector subbundle
$\tau_{D}^{-1}(U) \to U$ (respectively,
$\tau_{\tilde{D}^{\perp}}^{-1}(U) \to U$). In other words, we have
a local basis of sections adapted to the decomposition $E=D \oplus
\tilde D^{\perp}$. We will denote by $(q^i, v^{\a}) = (q^i, v^{a},
v^A)$ the corresponding local coordinates on $E$. The equations
defining $D$ as a vector subbundle of $E$ are:
\[
v^A = 0, \; \; \; \mbox{ for all } A.
\]
Thus, the equations of motion (\ref{gnhmotion}) are given by
\begin{equation} \left\{ \begin{array}{l} \dot v^c =  v^a v^b \widetilde C_{bc}^a -
(\rho^r_D)_c^i \frac{\partial V}{\partial q^i} \\
\dot q^i = (\rho^l_D)_a^i v^a.
\end{array} \right. \label{gnhcoord}
\end{equation} where the functions $\widetilde C_{ab}^c, (\rho^l_D)_b^i$ and $(\rho^r_D)_b^i$
are properly deduced in Appendix B.

Based on the compatibility condition, it seems natural to consider
some decompositions of the original vector bundle $E$. In
particular, we will use
\[
E=D\oplus D^{\perp} \hbox{  and  } E = \tilde{D} \oplus D^{\perp}
\]
with associated projectors:
 \begin{equation}P: D \oplus
D^{\perp} \rightarrow D \ \ \mbox{and} \ \ \Pi: \tilde{D} \oplus
D^{\perp} \rightarrow \tilde D \label{ProjP} \end{equation}
respectively, and the correspondent inclusions $i_D : D \rightarrow
E$ and $i_{\widetilde{D}} : \tilde{D} \rightarrow E$.

Next, denote by $(\lcf \cdot, \cdot \rcf_{E}, \rho_{E})$ the Lie
algebroid structure on the vector bundle $\tau_{E}: E \rightarrow
Q$. Then the bracket on $D$ given by
\begin{equation}
B_D(\sigma,\sigma'):= P(\lcf i_{D}(\sigma)\, , \,
i_{\widetilde{D}} \circ \Pi (\sigma')\rcf_E ), \; \; \mbox{ for }
\sigma, \sigma' \in \Gamma(\tau_{D})  \label{GNHbracket}
\end{equation} and the anchors maps
\begin{eqnarray}
\rho^l_{D} & := & \rho_E \circ i_D \label{GNHanchorl}\\
\rho^r_{D} & := & \rho_E \circ i_{\widetilde D} \circ \Pi
\label{GNHanchorr}
\end{eqnarray} define an algebroid structure $(B_D, \rho^l_{D} , \rho^r_{D})$
on the vector bundle $\tau_D : D \rightarrow Q$ with local
structure functions $\widetilde C_{ab}^c$, $(\rho^l_{D})_a^i$ and
$(\rho^r_{D})^i_a$ (see Appendix B). Note that, in general,
$\widetilde{C}_{ab}^c\not = - \widetilde{C}_{ba}^c$.

On the other hand, if we denote by $(q^i, p_{a})$ the
corresponding local coordinates on $D^*$ then the linear bracket
$\{\; , \; \}_{\Pi_{D^*}}$ on $D^*$ is determined by the following
relations:
\[
\{q^i\, , q^j\}_{\Pi_{D^*}}=0,\qquad \{q^i\, , p_a\}_{\Pi_{D^*}}=
(\rho^r_{D})^i_a, \qquad \{p_a\, ,
q^i\}_{\Pi_{D^*}}=-(\rho^l_{D})^i_a\qquad  \{p_a\, ,
p_b\}_{\Pi_{D^*}}=-p_c \widetilde{C}^c_{ab}\; .
 \]

Now, we take the Hamiltonian function $H_{E^*}: E^* \rightarrow
\mathbb{R}$ (the Hamiltonian energy) defined by
\[
H_{E^*}(\kappa_q) = \displaystyle \frac{1}{2} <\kappa_q, \kappa_q
>_{\mathcal G} + V(q), \; \; \; \mbox{ for } \kappa_{q} \in E^*_q
\]
where $< \cdot, \cdot >_{\mathcal G}$ is the bundle metric on
$\tau_{E^*}: E^* \rightarrow \mathbb{R}$ induced by ${\mathcal
G}$. Then, \emph{the constrained Hamiltonian function} $H_{D^*}:
D^* \rightarrow \mathbb{R}$ is given by
\[
H_{D^*} = H_{E^*} \circ P^*.
\]
Thus,
\begin{equation} \HH_{H_{D^*}}^{\Pi_{D^*}} = {(\rho_D^l)}_{a}^i
\frac{\partial H_{D^*}}{\partial p_{a}}\frac{\partial}{\partial
q^i}  - \left( {(\rho_D^r)}_{b}^j \frac{\partial H_{D^*}}{\partial
q^j} - \tilde{C}_{ab}^c p_{c} \frac{\partial H_{D^*}}{\partial
p_{a}}  \right) \frac{\partial}{\partial p_{b}}, \label{LocHvf-2}
\end{equation} which implies that the local expression of the
Hamilton equations is
$$\frac{dq^i}{dt} = {(\rho_D^l)}_{a}^i \frac{\partial
H_{D^*}}{\partial p_{a}}, \ \ \ \ \ \frac{dp_{b}}{dt} = - \left\{
{(\rho_D^r)}_{b}^j \frac{\partial H_{D^*}}{\partial q^j} -
\tilde{C}^c_{ab} p_{c} \frac{\partial H_{D^*}}{\partial p_{a}}
\right\}.$$

 In the induced local coordinates:
\[
H_{D^*}(q^i, p_a)=\frac{1}{2}\sum_{a=1}^m p_a^2+V(q^i),
\]
and the Hamilton equations are
 \begin{equation}\label{21'}
\displaystyle \frac{dq^i}{dt} = \sum_{a=1}^m(\rho_D^l)_a^i p_a, \;
\; \; \frac{dp_{b}}{dt} = - \left((\rho_D^r)^i_{a} \frac{\partial
V}{\partial q^i} - \sum_{b=1}^m\tilde{C}_{ab}^{c}p_{c} p_a\right).
\end{equation}

\begin{remark}{\rm
Note that if $\hbox{Leg}_{(L, D)}: D \rightarrow D^*$ is the
constrained Legendre transformation, that is,
\[
\hbox{Leg}_{(L, D)} = i_{D}^* \circ \flat_{\mathcal G} \circ i_{D}
\]
then, as the nonholonomic case (see Section \ref{qaz}), we deduce
that a curve $\gamma: I \rightarrow D$ is a solution of the motion
equations (\ref{gnhcoord}) if and only if the curve
$\hbox{Leg}_{(L, D)} \circ \gamma: I \rightarrow D^*$ is a
solution of the Hamilton equations (\ref{21'}). Thus, Eqs.
(\ref{21'}) may be considered as \emph{the generalized
nonholonomic Hamilton equations} for the generalized linear
nonholonomic system $(L, D, \tilde{D})$. }
\end{remark}

\ \

\noindent 2.3.2. {\bf Lagrangian mechanics for modifications of the
standard Lie bracket.}

\medskip

Now, we analyze another example that is of our interest since it
has a non skew-symmetric bracket. Consider the case of the
standard tangent bundle $\tau_{TQ}: TQ \to Q$, a Lagrangian $L$ of
mechanical type
\[
L(v)= \displaystyle \frac{1}{2} {\mathcal G}(v, v) -
V(\tau_{TQ}(v)), \; \; \; \mbox{ for } v\in TQ,
\]
and an arbitrary $(1,2)$-tensor field $T$:
\[
T: {\mathfrak X}(Q)\times {\mathfrak X}(Q)\to {\mathfrak X}(Q)
\]
It is easy to show that if we modify the standard Lie bracket
$[\cdot, \cdot]$ on ${\mathfrak X}(Q)$ by
\[
B_{TQ}(X, Y)=[X, Y]+ T(X,Y)
\]
then $(TQ,B_{TQ}, id_{TQ}, id_{TQ})$ is an algebroid.

If we take local coordinates $(q^i)$ and
\[
T(\frac{\partial}{\partial q^i}, \frac{\partial}{\partial
q^j})=T_{ij}^k \frac{\partial}{\partial q^k}
\]
then
\[
B_{TQ}(\frac{\partial}{\partial q^i}, \frac{\partial}{\partial
q^j})=T_{ij}^k \frac{\partial}{\partial q^k}
\]
and, therefore, $(B_{TQ})_{ij}^k=T_{ij}^k$. Thus, the linear bracket
$\{\cdot, \cdot\}_{\Pi_{T^*Q}}$ on $T^*Q$ is characterized by the
following relations
\[
\{q^i , q^j\}_{\Pi_{T^*Q}}=0,\qquad \{q^i,
p_j\}_{\Pi_{T^*Q}}=-\{p_i , q^j\}_{\Pi_{T^*Q}}=\delta_i^j, \qquad
 \{p_i , p_j\}_{\Pi_{D^*}}=-p_k T^k_{ij}\; .
 \]

In this case, the Hamiltonian function $H_{T^*Q}: T^*Q \rightarrow
\mathbb{R}$ is given by
\[
H_{T^*Q}(\kappa) = \displaystyle \frac{1}{2} <\kappa, \kappa
>_{\mathcal G} + V(\tau_{T^*Q}(\kappa)), \; \; \; \mbox{ for }
\kappa \in T^*Q.
\]
Thus, if ${\mathcal G} = {\mathcal G}_{ij} dq^i \otimes dq^j$ it
follows that
\[
H_{T^*Q}(q^i, p_i)=\frac{1}{2}{\mathcal G}^{ij}(q)p_ip_j+V(q)
\]
and then the associated Hamilton equations are
 \begin{equation}\label{H-modif}
\displaystyle \frac{dq^i}{dt} = {\mathcal G}^{ik}p_k, \; \; \;
\frac{dp_{j}}{dt} = - \frac{1}{2}\frac{\partial {\mathcal
G}^{ik}}{\partial q^j}p_ip_k-\frac{\partial V}{\partial q^j} -
T_{ij}^k{\mathcal G}^{il}p_kp_l.
\end{equation}
Note that $\{ H_{T^*Q}, H_{T^*Q}\}_{\Pi_{T^*Q}}=-T^k_{ij}{\mathcal
G}^{il}{\mathcal G}^{jm}p_kp_lp_m$ and the dynamics has in general a
dissipative behavior. An interesting case, is when  the tensor field
$T$ is skew-symmetric $(T(X,Y)=-T(Y,X)$ for all $X, Y\in {\mathfrak
X}(Q)$) then the hamiltonian $H_{T^*Q}$ is preserved,
$dH_{T^*Q}/dt=0$ along the flow. An important example, when this
condition is fulfilled, is the following one. Consider, as above, a
riemannian manifold $(Q, {\mathcal G})$ and an arbitrary affine
connection $\nabla$. Take the $(1,2)$ tensor field $S$ which encodes
the difference between it and the Levi-Civita connection
corresponding to the riemannian metric, that is,
 \[
 \nabla_X Y=\nabla^{\mathcal G}_XY+S(X,Y).
\]
This tensor field is called the \emph{contorsion tensor field}
(see \cite{Schou}; and also \cite{CCMD,Cort}).

Now, consider as $T(X,Y)=S(X,Y)-S(Y,X)$, and the  bracket of vector
fields:
\begin{equation}\label{nabla-modif}
B_{TQ}(X, Y)=\nabla_X Y-\nabla_Y X=[X, Y]+ S(X,Y)-S(Y,X).
\end{equation}
We obtain (\ref{H-modif}) but now the flow preserves the
Hamiltonian function. Equations (\ref{H-modif}) are important in
the modellization of generalized Chaplygin systems
(\cite{CCMD,Cort} and references therein), where now the
connection $\nabla$ is a metric connection with torsion.

\bigskip

\section{Exact symplectic algebroids and Hamiltonian vector fields}\label{section-exact}
\setcounter{equation}{0}

 In this section we will introduce the
notion of an {\it exact symplectic algebroid} and we will prove that
for such an algebroid $\tau_{{\bar E}}: {\bar E} \rightarrow {\bar Q}$, if $F$ is a real
$C^{\infty}$-function on the base manifold ${\bar Q}$, then $F$ induces a
Hamiltonian vector field on ${\bar Q}$.

First, we will give the definitions of the two differentials of a
real $C^{\infty}$-function $F$ on the base manifold ${\bar Q}$ of an
arbitrary vector bundle $\tau_{\bar E} : {\bar E} \rightarrow {\bar Q}$ with an algebroid
structure $(B_{\bar E}, \rho^l_{\bar E}, \rho^r_{\bar E})$. We will also give the
definition of the differential of a section of the vector bundle
$\tau_{{\bar E}^*} : {\bar E}^* \rightarrow {\bar Q}$.

In fact, the \emph{left differential $d^l_{\bar E}F$ of $F$} is given by
\begin{equation}\label{23'}
(d_{\bar E}^lF)(\sigma) = \rho_{\bar E}^l(\sigma)(F),
\end{equation}
and the \emph{right differential} $d^r_{\bar E}F$ of $F$ is
\begin{equation}\label{23''}
(d_{\bar E}^rF)(\sigma) = \rho_{\bar E}^r(\sigma)(F),
\end{equation}
for $\sigma \in \Gamma(\tau_{\bar E}).$

On the other hand, if $\kappa : {\bar Q} \rightarrow {\bar E}^*$ is a section of
the dual vector bundle $\tau_{{\bar E}^*} : {\bar E}^* \rightarrow {\bar Q}$ then
\emph{the differential of $\kappa$} is the section $d^{lr}_{\bar E}
\kappa$ of the vector bundle $\tau_{\otimes_2^0 {\bar E}^*} : \otimes_2^0
{\bar E}^* \rightarrow {\bar Q}$ defined by \begin{equation} (d^{lr}_{\bar E}
\kappa)(\sigma, \sigma') = \rho_{\bar E}^l(\sigma)(\kappa (\sigma')) -
\rho_{\bar E}^r(\sigma')(\kappa(\sigma)) - \kappa (B_{\bar E}(\sigma, \sigma'))
\label{Difflr} \end{equation} for $ \sigma, \sigma' \in
\Gamma(\tau_{\bar E})$. Note that, from (\ref{Eqfun}), it follows that
$d^{lr}_{\bar E} \kappa \in \Gamma(\tau_{\otimes_2^0 {\bar E}^*}).$

If $(q^i)$ are local coordinates on ${\bar Q}$, $\{\sigma_{\a} \}$ is a
local basis of $\Gamma(\tau_{\bar E})$, $\{\sigma^{\a} \}$ is the dual
basis of $\Gamma(\tau_{{\bar E}^*})$ and $\kappa = \kappa_{\gamma} \sigma^{\gamma}$
then
$$d_{\bar E}^lF = {(\rho_{\bar E}^l)}_{\a}^i \frac{\partial F}{\partial q^i} \sigma^{\a}, \ \ \ \ \
d_{\bar E}^rF = {(\rho_{\bar E}^r)}_{\a}^i \frac{\partial F}{\partial q^i} \sigma^{\a} $$
$$d_{\bar E}^{lr} \kappa = \left\{ {(\rho_{\bar E}^l)}_{\beta}^i \frac{\partial \kappa _{\gamma}}{\partial q^i}
 -  {(\rho_{\bar E}^r)}_{\gamma}^i \frac{\partial \kappa_{\beta}}{\partial q^i}
  +  {(B_{\bar E})}_{\beta \gamma}^{\mu} \kappa_{\mu}  \right\} \sigma^{\beta} \otimes \sigma^{\gamma},$$
where ${(B_{\bar E})}_{\beta \gamma}^{\mu}, {(\rho_{\bar E}^l)}_{\beta}^i$ and
${(\rho_{\bar E}^r)}_{\gamma}^i$ are the local structure functions of
$\tau_{\bar E}: {\bar E} \rightarrow {\bar Q}$ with respect to the local coordinates
$(q^i)$ and the local basis $\{\sigma_{\a} \}$. Furthermore, if $F
\in C^{\infty}({\bar Q})$ and $\kappa \in \Gamma(\tau_{{\bar E}^*})$, we have that
$$d^{lr}_{\bar E}(F \kappa) = F d^{lr}_{\bar E} \kappa + d_{\bar E}^l F \otimes \kappa - \kappa \otimes
d^r_{\bar E} F.$$

\begin{definition} A vector bundle $\tau_{\bar E} : \bar E \rightarrow \bar Q$
with an algebroid structure $(B_{\bar{E}}, \rho^l_{\bar{E}},
\rho^r_{\bar{E}})$ is said to be exact  symplectic if there exists
$\lambda_{\bar{E}} \in \Gamma(\tau_{\bar{E}^*})$ such that the
tensor of type $(0,2)$ \ $\Omega_{\bar E}$ defined by
$\Omega_{\bar{E}} = - d^{lr}_{{\bar{E}}} \lambda_{{\bar{E}}}\in
\Gamma(\tau_{\otimes_2^0 ({\bar{E}})^*})$ is skew-symmetric and
nondegenerate.
\end{definition}

Thus, if $\tau_{{\bar{E}}} : {\bar{E}} \rightarrow {\bar{Q}}$ is an
exact symplectic algebroid then the map
$\flat^{l}_{\Omega_{{\bar{E}}}} : \Gamma(\tau_{\bar{E}}) \rightarrow
\Gamma(\tau_{\bar{E}^*})$ given by
$$ \flat^l_{\Omega_{{\bar{E}}}}(X) = i_X \Omega_{{\bar{E}}}=X\lrcorner\, \Omega_{{\bar{E}}},
\ \ \ \mbox{for } X \in \Gamma(\tau_{\bar{E}})$$
 is an isomorphism of $C^{\infty}({\bar{Q}})$-modules. Therefore, for a real $C^{\infty}$-function
 ${\bar{H}}$ on ${\bar{Q}}$ (\emph{a Hamiltonian function}) we may consider the
 \emph{right Hamiltonian section}
 $\HH^{({\Omega_{{\bar{E}}},r})}_{{\bar{H}}}$ of ${\bar{H}}$ defined by
\begin{equation}\label{24'}
 \HH^{({\Omega_{{\bar{E}}},r})}_{{\bar{H}}} = (\flat^l_{\Omega_{{\bar{E}}}})^{-1}(d^r_{{\bar{E}}} {\bar{H}})\in
 \Gamma(\tau_{\bar{E}})
 \end{equation}
 and the \emph{left-right Hamiltonian vector field} $\HH^{({\Omega_{{\bar{E}}},lr})}_{{\bar{H}}}$ of ${\bar{H}}$ on $\bar{Q}$ given by
  $$\HH^{({\Omega_{{\bar{E}}},lr})}_{{\bar{H}}} = \rho_{\bar E}^l \left( \HH^{({\Omega_{{\bar{E}}},r})}_{{\bar{H}}} \right)\in {\mathfrak X}(\bar{Q}).$$

The integral curves of
$\HH^{({\Omega_{{\bar{E}}},lr})}_{{\bar{H}}}$ are called the
solutions of the \emph{Hamilton equations} for ${\bar{H}}$.

\bigskip

\section{An exact symplectic formulation of the Hamiltonian dynamics on an
algebroid}\label{section-main} \setcounter{equation}{0}

 In this
section, we will propose an exact symplectic formulation of the
Hamiltonian dynamics on an algebroid. First, we will see how the
bracket of sections of an arbitrary algebroid structure $(B_E,
\rho_E^l, \rho_E^r)$ on a vector bundle $\tau_{E}: E \to Q$ may be
written in terms of a suitable $\rho_E^l$-connection and a
$\rho_E^r$-connection (for the definition of a
$\rho_{E}^l$-connection and a $\rho_{E}^r$-connection, see Appendix
A). More precisely, we will prove the following result:

\begin{prop} \label{PropBracket}

\begin{enumerate}

\item Let $\tau_E : E \rightarrow Q$ be a vector bundle and $\rho^l_E : E \rightarrow TQ$,
$\rho^r_E : E \rightarrow TQ$ two anchor maps. If $D^l :
\Gamma(\tau_E)  \times \Gamma(\tau_E) \rightarrow \Gamma(\tau_E)$
(respectively, $D^r: \Gamma(\tau_E) \times \Gamma(\tau_E)
\rightarrow \Gamma(\tau_E)$) is a $\rho^l_E$-connection
(respectively, $\rho_E^r$-connection) on $\tau_E : E \rightarrow Q$
then $(B_E, \rho_E^l, \rho_E^r)$ is an algebroid structure on
$\tau_E : E \rightarrow Q$, where $B_E: \Gamma(\tau_E) \times
\Gamma(\tau_E) \rightarrow \Gamma(\tau_E)$ is the bracket defined by
    \begin{equation} B_E(\sigma, \sigma') = D^l_{\sigma} \sigma' - D^r_{\sigma'} \sigma \ \ \
    \mbox{for} \  \sigma, \sigma' \in \Gamma(\tau_E). \label{BE} \end{equation}

\item Let $(B_E, \rho_E^l, \rho_E^r)$ be an algebroid structure on a real vector
bundle $\tau_E : E \rightarrow Q$. Then, there exists a $\rho^l_E$-connection $D^l$ (respectively,
$\rho_E^r$-connection $D^r$) such that $$B_E(\sigma, \sigma') = D^l_{\sigma} \sigma'
- D^r_{\sigma'} \sigma \ \ \ \mbox{for} \  \sigma, \sigma' \in \Gamma(\tau_E).$$
\end{enumerate}
\end{prop}

\begin{proof}
\ \
\begin{enumerate} \item From (\ref{Condcon}) (in Appendix A) and (\ref{BE}), it follows
that $(B_E, \rho_E^l, \rho_E^r)$ is an algebroid structure on the
vector bundle $\tau_E :E \rightarrow Q$.
\item Suppose that $\tilde D^l : \Gamma(\tau_E)  \times \Gamma(\tau_E) \rightarrow \Gamma(\tau_E)$
(respectively, $\tilde D^r: \Gamma(\tau_E) \times \Gamma(\tau_E)
\rightarrow \Gamma(\tau_E)$) is an arbitrary $\rho^l_E$-connection
(respectively, $\rho_E^r$-connection) on $\tau_E : E \rightarrow
Q$. Note that the vector bundle $\tau_{E}: E \rightarrow Q$ admits
a $\rho_{E}^{l}$-connection and a $\rho_{E}^{r}$-connection (see
Remark \ref{remRhoconn} in Appendix A). Then, we may introduce the
map $T: \Gamma(\tau_E)  \times \Gamma(\tau_E) \rightarrow
\Gamma(\tau_E)$ given by $$T(\sigma, \sigma') = B_E(\sigma,
\sigma') - \tilde D^l_{\sigma} \sigma' + \tilde D^r_{\sigma'}
\sigma \ \ \ \mbox{for} \  \sigma, \sigma' \in \Gamma(\tau_E).$$

   It is easy to check that $T$ is a section of the vector bundle $\tau_{\otimes_2^0 E^*} : \otimes_2^0 E^* \rightarrow Q$. Thus,
$$B_E(\sigma, \sigma') = D^l_{\sigma} \sigma' - D^r_{\sigma'} \sigma \ \ \ \mbox{for} \  \sigma, \sigma' \in \Gamma(\tau_E),$$ where $D^l$ and $D^r$ are the $\rho^l_E$-connection and the $\rho_E^r$-connection, respectively,
defined by
$$ D^l_{\sigma} \sigma' = \tilde D^l_{\sigma} \sigma' \ \ \  D^r_{\sigma'} \sigma = \tilde D^r_{\sigma'} \sigma - T(\sigma, \sigma').$$
\end{enumerate}
\end{proof}

Let $\tau_E : E \rightarrow Q$ be an algebroid with structure $(B_E,
\rho^l_E, \rho^r_E),$ $H:E^* \rightarrow \R$ be a Hamiltonian
function on $E^*$ and $\HH^{\Pi_{E^*}}_H$ be the corresponding
Hamiltonian vector field on $E^*$ (see section
\ref{subSecAlmosLeibniz}).

Let us, also, consider the vector bundle $\tau_{\TT_l^EE^*} :
\TT_l^EE^* \rightarrow E^*$ over $E^*$ whose fiber at the point
$\kappa_q \in E^*_q$ is $$(\TT_l^EE^*)_{\kappa_q} = \{ (\sigma_q, \tilde
X_{\kappa_q}) \in E_q \times T_{\kappa_q}E^* \ / \ \rho^l_E(\sigma_q) =
(T_{\kappa_q} \tau_{E^*})(\tilde X_{\kappa_q}) \}.$$

$\TT_l^EE^*$ is called the \emph{the left $E$-tangent bundle to
$E^*$} (see \cite{LMM,Medina} and references therein).

Now, we will prove the two main results of our paper.

\begin{theo} \label{TeoExact}
The vector bundle $\tau_{\TT_l^EE^*} : \TT_l^EE^* \rightarrow E^*$
admits an exact symplectic algebroid structure.
\end{theo}

\begin{theo} \label{TeoHamiltonian}
If $\Omega_{\TT_l^EE^*}$ is the exact symplectic structure on the
algebroid $\tau_{\TT_l^EE^*} : \TT_l^EE^* \rightarrow E^*$ then the
left-right Hamiltonian vector field of a hamiltonian function $H:
E^*\to \R$ with respect to $\Omega_{\TT_l^EE^*}$ is just
$\HH^{\Pi_{E^*}}_H$, that is,
\[
\HH^{({\Omega_{{\TT_l^EE^*}},lr)}}_{H}=\HH^{\Pi_{E^*}}_H.
\]

\end{theo}

\begin{proof}

[Theorem  \ref{TeoExact}]

Using Proposition \ref{PropBracket}, we deduce that there exists a
$\rho_{E}^{l}$-connection $D^{l}$ and a $\rho_{E}^r$-connection
$D^r$ on $\tau_{E}: E \rightarrow Q$ such that
$$B_E(\sigma,
\bar \sigma) = D^l_{\sigma} \bar \sigma - D^r_{\bar \sigma} \sigma ,
\ \ \ \ \ \ \mbox{for } \sigma, \bar \sigma \in
\Gamma(\tau_E).$$

Denote by $$^{D^l{\mathbf h}}_{\kappa_q} : E_q \rightarrow
T_{\kappa_q} E^*, \hspace*{4cm} ^{D^r{\mathbf h}}_{\kappa_q} : E_q
\rightarrow T_{\kappa_q} E^*$$ the $(D^l)$-horizontal lift and the
$(D^r)$-horizontal lift, respectively, at the point $\kappa_q \in
E^*_q$ (see Appendix A, for a detailed description of horizontal
and vertical lifts).

Thus, if $\sigma \in \Gamma (\tau_E)$ we can consider the
corresponding $(D^l)$-horizontal lift $\sigma^{(D^l){\mathbf h}}
\in {\mathfrak X}(E^*)$ (respectively, the corresponding
$(D^r)$-horizontal lift $\sigma^{(D^r){\mathbf h}} \in {\mathfrak
X}(E^*)$) given by
\[
\sigma^{(D^l){\mathbf h}}(\kappa_q)=(\sigma(q))^{D^l{\mathbf
h}}_{\kappa_q}\quad \hbox{   and   }\quad \sigma^{(D^r){\mathbf
h}}(\kappa_q)=(\sigma(q))^{D^r{\mathbf h}}_{\kappa_q}, \; \;
\mbox{ for } \kappa_q \in E^*_q.
\]
 Now, we introduce the  $D^l$-horizontal lift of $\sigma$ to
 $\Gamma(\tau_{\TT_{l}^EE^*})$ as the section of $\tau_{\TT_{l}^EE^*}:
 {\TT_{l}^EE^*} \rightarrow E^*$ defined by
$$\sigma^{(D^l){\mathbf h}}_l(\kappa_q) = (\sigma(q), \sigma^{(D^l){\mathbf h}}(\kappa_q)),
\ \ \ \ \ \ \mbox{for } \kappa_q \in E^*_q.$$ Note that the vector
field $\sigma^{(D^l)h}$ is $\tau_{E^*}$-projectable on the vector
field $\rho_{E}^l(\sigma) \in {\mathfrak X}(Q)$ (see Appendix A)
and, therefore, $\sigma_{l}^{(D^l){\mathbf h}}(\kappa_{q}) \in
(\TT_{l}^{E}E^*)_{\kappa_{q}}$. Moreover, from (\ref{A.38'}), it
follows that
\begin{equation}\label{25'}
\sigma^{(D^l){\mathbf h}}_{l} + \bar{\sigma}_{l}^{(D^l){\mathbf
h}} = (\sigma + \bar{\sigma})_{l}^{(D^l){\mathbf h}}, \; \; \;
(f\sigma)_{l}^{(D^l){\mathbf h}} = (f \circ
\tau_{E^*})\sigma_{l}^{(D^l){\mathbf h}},
\end{equation}

On the other hand, if $\nu \in \Gamma(\tau_{E^*})$ then the
vertical lift $\nu^{\mathbf v}_l \in \Gamma(\TT_l^EE^*)$ of $\nu$
to $\TT_l^EE^*$ is given by  $$\nu^{\mathbf v}_l(\kappa_q) = (0,
\nu^{\mathbf v}(\kappa_q)), \ \ \ \ \ \ \mbox{for } \kappa_q \in
E^*_q,$$ where $\nu^{\mathbf v} \in {\mathfrak X}(E^*)$ is the
standard vertical lift of $\nu$ (see Appendix A). In this case, we have also that
(see (\ref{A.39'}))
\begin{equation}\label{25''}
\nu^{\mathbf v}_{l} + \bar{\nu}^{\mathbf v}_{l} = (\nu +
\bar{\nu})^{\mathbf v}_{l}, \; \; \; (f\nu)^{\mathbf v}_{l} = (f
\circ \tau_{E^*}) \nu^{\mathbf v}_{l},
\end{equation}
for $\nu, \bar{\nu} \in \Gamma(\tau_{E^*})$ and $f \in
C^{\infty}(Q)$.

 It is clear that if $\{\sigma_{\a}\}$ is a local basis of
$\Gamma(\tau_E)$ and $\{\sigma ^{\a}\}$ is the dual basis fo
$\Gamma(\tau_{E^*})$,  then $\{(\sigma_{\a})^{(D^l){\mathbf h}}_l,
(\sigma^{\a})^{\mathbf v}_l \}$ is a local basis of
$\Gamma(\TT_l^EE^*)$.

Next, let $R$ be a tensor of type $(1, 1)$ on $\tau_{E}: E
\rightarrow Q$, that is, $R: \Gamma(\tau_{E}) \rightarrow
\Gamma(\tau_{E})$ is a $C^{\infty}(Q)$-linear map. Then, there
exists a unique vector field $R^{\mathbf v}$ on $E^*$, \emph{the
vertical lift of R}, such that
\[
R^{\mathbf v} (f \circ \tau_{E^*}) = 0 \ \ \mbox{and} \ \ R^{\mathbf
v}(\widehat \gamma) = \widehat{R( \gamma)}
\]
for $f \in C^{\infty}(Q)$ and $\gamma \in \Gamma(\tau_E)$. Here, \,
$\widehat{}$ \, denotes the linear function on $E^*$ induced by a
section in $\Gamma(\tau_E)$.

The vertical lift of the tensor $R$ to $\TT_l^EE^*$ is denoted by
$R^{\mathbf v}_l$ and given by
$$R^{\mathbf v}_l(\kappa_q) = (0, R^{\mathbf v}(\kappa_q)),  \ \ \ \ \ \ \mbox{for } \kappa_q \in E^*_q.$$

Now, we will define the algebroid structure on $\TT_l^EE^*$. First,
consider the two anchor maps $\rho^l_{\TT_l^EE^*} : \TT_l^EE^*
\rightarrow TE^*$ and $\rho^r_{\TT_l^EE^*} : \TT_l^EE^* \rightarrow
TE^*$ given by
\begin{eqnarray*} \rho^l_{\TT_l^EE^*} (\sigma_q, \tilde X_{\kappa_q}) &=& \tilde X_{\kappa_q},
 \\ \rho^r_{\TT_l^EE^*} (\sigma_q, \tilde X_{\kappa_q}) &=& \tilde X_{\kappa_q} -
 (\sigma_q)^{(D^l){\mathbf h}}_{\kappa_q} + (\sigma_q)^{(D^r){\mathbf h}}_{\kappa_q}, \end{eqnarray*}
for $(\sigma_q, \tilde X_{\kappa_q}) \in (\TT_l^EE^*)_{\kappa_q}$,
with $\kappa_q \in E^*_q$.

Note that
\begin{equation}
\begin{array}{rr}
 \rho^l_{\TT_l^EE^*} (\sigma^{(D^l){\mathbf h}}_l) =
\sigma^{(D^l){\mathbf h}}, \ \ \ \ \ &   \ \ \ \ \ \rho^l_{\TT_l^EE^*}(\nu^{\mathbf v}_l) = \nu^{\mathbf v},  \\
\rho^r_{\TT_l^EE^*} (\sigma^{(D^l){\mathbf h}}_l) =
\sigma^{(D^r){\mathbf h}}, \ \ \ \ \ &   \ \ \ \ \
\rho^r_{\TT_l^EE^*}(\nu^{\mathbf v}_l) = \nu^{\mathbf v},
\end{array} \label{Anclas}
\end{equation}
 for $ \sigma \in \Gamma(\tau_E)$ and $\nu \in \Gamma(\tau_{E^*}).$


Next, we define a bracket $B_{\TT_{l}^EE^*}$ on the space
$\Gamma(\tau_{\TT^{E}_{l}E^*})$ by lifting the bracket $B_{E}$ on
$\Gamma(\tau_{E})$.

Let $R$ be a tensor of type $(1,3)$ on $\tau_{E}: E \rightarrow
Q$, that is, $R: \Gamma(\tau_{E}) \times \Gamma(\tau_{E}) \times
\Gamma(\tau_{E}) \rightarrow \Gamma(\tau_{E})$ is a
$C^{\infty}(Q)$-linear map. Then, using (\ref{Eqfun}),
(\ref{25'}), (\ref{25''}) and (\ref{Condcon}), we deduce that
there exists a unique bracket $B_{\TT_{l}^EE^*}$ on
$\Gamma(\TT_{l}^EE^*)$ such that
\begin{equation}
\begin{array}{l}
B_{\TT_l^EE^*} (\sigma^{(D^l){\mathbf h}}_l , \bar
\sigma^{(D^l){\mathbf h}}_l) = B_E( \sigma , \bar
\sigma)^{(D^l){\mathbf h}}_l + R(\sigma, \bar \sigma,
\cdot)^{\mathbf v}_l  \\
B_{\TT_l^EE^*} (\sigma^{(D^l){\mathbf h}}_l , \nu_l^{\mathbf v}) = (D^l_{\sigma}  \nu)^{\mathbf v}_l ,  \\
B_{\TT_l^EE^*} (\kappa_l^{\mathbf v}, \bar \sigma^{(D^l){\mathbf h}}_l) = -(D^r_{\bar \sigma}  \kappa)^{\mathbf v}_l  \\
B_{\TT_l^EE^*}(\kappa_l^{\mathbf v}, \nu_l^{\mathbf v}) = 0
\end{array} \label{BracketEtE*}
\end{equation}
for $\sigma, \bar{\sigma} \in \Gamma(\tau_{E})$ and $\kappa, \nu \in
\Gamma(\tau_{E^*})$. Note that $R(\sigma, \bar{\sigma}, \cdot):
\Gamma(\tau_{E}) \rightarrow \Gamma(\tau_{E})$ is a tensor of type
$(1, 1)$ on $\tau_{E}: E \rightarrow Q$. Furthermore, from
(\ref{Anclas}) and (\ref{BracketEtE*}), it follows that
$(B_{\TT_l^EE^*}, \rho^l_{\TT_l^EE^*}, \rho^r_{\TT_l^EE^*})$ is an
algebroid structure on $\TT_l^EE^*$.

Now, we will endow the algebroid $(\TT_{l}^EE^*, B_{\TT_l^EE^*},
\rho^l_{\TT_l^EE^*}, \rho^r_{\TT_l^EE^*})$ with an exact symplectic
structure. In fact, the dual vector bundle $\tau_{(\TT_l^EE^*)^*} :
(\TT_l^EE^*)^* \rightarrow E^*$ admits a canonical section
$\lambda_{\TT_l^EE^*}$, the \emph{Liouville section}, defined by
$$\lambda_{\TT_l^E E^*} (\kappa_q) (\sigma_q, \tilde X_{\kappa_q}) =
\kappa_q(\sigma_q),$$ for $\kappa_q \in E^*_q$ and $(\sigma_q, \tilde
X_{\kappa_q}) \in (\TT_l^EE^*)_{\kappa_q}.$

Note that
\begin{equation}
\lambda_{\TT_l^EE^*}(\sigma^{(D^l){\mathbf h}}_l) =
\widehat{\sigma}, \ \ \ \lambda_{\TT_l^EE^*}(\kappa^{\mathbf v}_l)
=0 \label{Liouloc}
\end{equation}
for $\sigma \in \Gamma(\tau_E)$ and $\kappa \in
\Gamma(\tau_{E^*})$.

Consider the tensor of type $(0, 2)$ on $\tau_{{\mathcal
T}_{l}^EE^*}: \TT_{l}^EE^* \to E^*$ given by
\[
\Omega_{\TT_l^EE^*} = -d^{lr}_{\TT_l^EE^*} \left(
\lambda_{\TT_l^EE^*} \right).
\]

From (\ref{Difflr}), (\ref{Anclas}), (\ref{BracketEtE*}),
(\ref{Liouloc}) and (\ref{DefDh}), we deduce that
\begin{eqnarray}
& & \Omega_{\TT_l^EE^*}(\sigma^{(D^l){\mathbf h}}_l ,
\bar{\sigma}^{(D^l){\mathbf h}}_l) = - \sigma^{(D^l){\mathbf
h}}(\widehat{\bar{\sigma}}) + (\bar{\sigma})^{(D^r){\mathbf
h}}(\widehat{\sigma}) + \lambda_{\TT_l^EE^*} \left( B_E( \sigma ,
\bar{\sigma})^{(D^l){\mathbf h}}_l +
R(\sigma, \bar{\sigma}, \cdot)^{\mathbf v}_l \right) \nonumber \\
& & \hspace{3.55cm} = - \widehat{D^l_{\sigma} \bar{\sigma}} +
\widehat{D^r_{\bar{\sigma}}
\sigma} + \widehat{B_E( \sigma , \bar{\sigma})} = 0 \nonumber \\
& & \Omega_{\TT_l^EE^*} (\sigma^{(D^l){\mathbf h}}_l , \nu_l^{\mathbf
v}) = \nu^{\mathbf v}(\widehat{\sigma}) \circ \tau_E +
\lambda_{\TT_l^EE^*}
\left( (D^l_{\sigma} \nu)^{\mathbf v}_l \right) = \nu(\sigma) \circ \tau_{E^*}  \label{OmegaEtE*} \\
& & \Omega_{\TT_l^EE^*} (\kappa_l^{\mathbf v},
\bar{\sigma}^{(D^l){\mathbf h}}_l) = - \kappa^{\mathbf
v}(\widehat{\bar{\sigma}}) \circ \tau_E - \lambda_{\TT_l^EE^*}
\left( (D^r_{\bar{\sigma}} \kappa)^{\mathbf v}_l \right) = - \kappa(\bar{\sigma}) \circ \tau_{E^*} \nonumber \\
& & \Omega_{\TT_l^EE^*}(\kappa_l^{\mathbf v}, \nu_l^{\mathbf v}) =
0 \nonumber
\end{eqnarray}
 for $\sigma$, $ \bar{\sigma} \in \Gamma(\tau_E) $ and $\kappa$, $\nu \in \Gamma(\tau_{E^*})$.
 As a consequence, $\Omega_{\TT_l^EE^*}$ is a skew-symmetric tensor.

In addition, if $\{\sigma_{\a} \}$ is a local basis of
$\Gamma(\tau_E)$, we have that $\{ {(\sigma_{\a})}^{(D^l){\mathbf
h}}_l, {(\sigma^{\a})}^{\mathbf v}_l \}$ is a local basis of
$\Gamma(\tau_{\TT_l^EE^*})$ and
\begin{equation} \Omega_{\TT_l^EE^*} = {\left(
{(\sigma_{\a})}^{(D^l){\mathbf h}}_l \right) }^* \wedge { \left(
{(\sigma^{\a})}^{\mathbf v}_l \right) }^* \label{Explocalsym} \end{equation}
where $\left\{ {\left( {(\sigma_{\a})}^{(D^l){\mathbf h}}_l \right) }^* ,  {
\left( {(\sigma^{\a})}^{\mathbf v}_l \right) }^* \right\}$ is the dual basis
of $\{ {(\sigma_{\a})}^{(D^l){\mathbf h}}_l, {(\sigma^{\a})}^{\mathbf v}_l \}$.
Therefore, it is clear that $\Omega_{\TT_l^EE^*}$ is nondegenerate.

This ends the proof of our theorem.
\end{proof}

\begin{remark}{\rm
In the particular case when $E$ is an skew-symmetric algebroid (that
is, the bracket $B_{E}$ is skew-symmetric), the exact symplectic
structure $\Omega_{\TT_{l}^EE^*}$ was considered by Popescu et al
\cite{PoPo} in order to develop a symplectic description of the
Hamiltonian mechanics on skew-symmetric algebroids. In this case,
the exact symplectic structure $\Omega_{\TT_{l}^EE^*}$ does not
depend on the chosen connection.}
\end{remark}

Now, suppose that $E$ is a Lie algebroid with Lie algebroid
structure $(\lcf \cdot, \cdot \rcf_{E}, \rho_{E})$. Then,
$(\Gamma(\tau_E), \lcf \cdot, \rcf_E)$ is a real Lie algebra and
$$\lcf \sigma, f \sigma' \rcf_{E}  = f \lcf \sigma,  \sigma'
\rcf_{E} + \rho_{E}(\sigma)(f) \sigma', \;
$$ for  $\sigma, \sigma' \in \Gamma(\tau_E)$  and  $f \in C^{\infty}(Q)$.
Moreover, we have that the anchor map $\rho^l_E = \rho^r_E =
\rho_E$ is a Lie algebra morphism, i.e. $\rho_E \left( \lcf
\sigma, \sigma' \rcf_E \right) = [\rho_E(\sigma) ,
\rho_E(\sigma')]$, where $[\cdot, \cdot]$ is the standard bracket
of vector fields. In addition, it is well known that
$\TT_l^EE^*=\TT^EE^*$ also admits a Lie algebroid structure (see,
for instance, \cite{LMM}).

Now, we will see that our construction permits to recover this Lie
algebroid structure.

To do this, consider an arbitrary bundle metric $\GG: E \times_Q E
\rightarrow \R$ on $E$ and denote by $\nabla^{\GG}: \Gamma(E)
\times \Gamma(E) \rightarrow \Gamma(E)$ the {\it Levi Civita}
connection induced by $\GG$. Then, we have that
\begin{equation}\label{notorsion}
\lcf\sigma, \sigma' \rcf_E = \nabla^{\GG}_{\sigma} \sigma' -
{\nabla}^{\GG}_{\sigma'} \sigma, \; \;  \mbox{ for } \sigma,
\sigma' \in \Gamma(\tau_E).
\end{equation}
On the other hand \emph{ the curvature of the connection}
${\nabla}^\GG$  is the tensor field of type (1,3) on $\tau_{E}: E
\rightarrow Q$
$$R^{\nabla^{\GG}}: \Gamma(\tau_E) \times \Gamma(\tau_E)
\times \Gamma(\tau_E) \rightarrow \Gamma(\tau_E)$$ defined for each
$\sigma, \sigma', \sigma'' \in \Gamma(\tau_E)$ as
$$R^{\nabla^{\GG}} (\sigma, \sigma',  \sigma'') =
\nabla^{\GG}_{\sigma} (\nabla^{\GG}_{\sigma'} \sigma'') -
\nabla^{\GG}_{\sigma'} (\nabla^{\GG}_{\sigma} \sigma'') -
\nabla^{\GG}_{\lcf \sigma, \sigma' \rcf_E} \sigma''.$$

Using (\ref{DefDh}) and the fact that $\rho_{E}$ is a Lie algebra
morphism, it is easy to prove that
$$ [\sigma^{\nabla^{\GG}{\mathbf h}}, (\sigma')^{\nabla^{\GG}{\mathbf h}}]
(f \circ \tau_{E^*}) = \left(\rho_E  \lcf \sigma, \sigma' \rcf_E \right)(f) \circ \tau_{E^*} $$
and
$$[\sigma^{\nabla^{\GG}{\mathbf h}}, (\sigma')^{\nabla^{\GG}{\mathbf h}}] (\widehat{\sigma''}) =
\left( \lcf \sigma, \sigma' \rcf_E \right)^{\nabla^{\GG}{\mathbf
h}} (\widehat{\sigma''}) + R^{\nabla^{\GG}}(\sigma, \sigma',
\cdot)^{\mathbf v}(\widehat{\sigma''}),
$$
for $\sigma, \sigma', \sigma'' \in \Gamma(\tau_{E})$, where
$\sigma^{\nabla^{\GG}{\mathbf h}}$ denotes the
$\nabla^{\GG}$-horizontal lift of $\sigma$ to ${\mathfrak
X}(E^*)$.
 Thus, we conclude that
$$[\sigma^{\nabla^{\GG}{\mathbf h}},
(\sigma')^{\nabla^{\GG}{\mathbf h}}] = \lcf \sigma, \sigma' \rcf_E
^{\nabla^{\GG}{\mathbf h}} + \left(R^{\nabla^{\GG}} (\sigma,
\sigma', \cdot) \right)^{\mathbf v}.$$ By a similar argument,
using (\ref{DefDh}) and (\ref{A.39.0}), we have that
$$[\sigma^{\nabla^{\GG}{\mathbf h}}, \nu^{\mathbf v}] =
( \nabla^{\GG}_{\sigma} \nu)^{\mathbf v} \ \ \ \ \ \mbox{and} \ \ \
\ \ [\kappa^{\mathbf v}, \nu^{\mathbf v}] = 0, \;\;\; \mbox{ for }
\kappa, \nu\in \Gamma(\tau_{E^*}).$$

Therefore, if we replace in (\ref{BracketEtE*}) the tensor $R$ by
the curvature $R^{\nabla^{\GG}}$ of the connection $\nabla^{\GG}$
then we obtain that
\begin{eqnarray}
& & B_{\TT^EE^*} (\sigma^{\nabla^{\GG}{\mathbf h}}_l ,
\sigma'^{\nabla^{\GG}{\mathbf h}}_l) = (\lcf \sigma ,
\sigma'\rcf_E , [\sigma^{\nabla^{\GG}{\mathbf h}},
(\sigma')^{\nabla^{\GG}{\mathbf h}}] ),
 \nonumber \\
& & B_{\TT^EE^*} (\sigma^{\nabla^{\GG}{\mathbf h}}_l , \nu_l^{\mathbf v}) =
(0, [{\sigma}^{\nabla^{\GG}{\mathbf h}},   \nu^{\mathbf v}])   =
- B_{\TT^EE^*} (\nu_l^{\mathbf v} , \sigma^{\nabla^{\GG}{\mathbf h}}_l ),\nonumber \\
& & B_{\TT^EE^*}(\kappa_l^{\mathbf v}, \nu_l^{\mathbf v}) = 0.
\nonumber
\end{eqnarray}
 Consequently, if $\sigma$ (respectively, $\sigma'$) is a section
 of $\tau_{E}: E \rightarrow Q$ and $X$ (respectively, $X'$) is a
 vector field on $E^*$ which is $\tau_{E^*}$-projectable on
 $\rho_{E}(\sigma)$ (respectively, $\rho_{E}(\sigma')$) then
 $(\sigma, X)$ and $(\sigma', X')$ are sections of
 $\tau_{\TT^EE^*}: \TT^{E}E^* \rightarrow E^*$ and
 \begin{equation}B_{\TT^EE^*} ((\sigma, X), (\sigma', X'))
= (\lcf \sigma, \sigma' \rcf_E, [X, X']). \label{UsualBracket}
\end{equation}
This implies that $B_{\TT^EE^*}$ is the canonical Lie bracket on
$\Gamma(\tau_{\TT^EE^*})$ (see \cite{LMM}).

 In Section \ref{Section5} we will use the following properties of the curvature of the connection $\nabla^{\GG}$:
\begin{equation}R^{\nabla^{\GG}}(\sigma, \sigma') \sigma'' = - R^{\nabla^{\GG}}(\sigma', \sigma)\sigma'' \label{CurvSkew} \end{equation} and
\begin{equation} R^{\nabla^{\GG}}(\sigma, \sigma') \sigma'' + R^{\nabla^{\GG}}(\sigma', \sigma'') \sigma + R^{\nabla^{\GG}}(\sigma'', \sigma) \sigma' =0
\qquad \hbox{(first Bianchi identity)}.\label{CurvBianchi}
\end{equation}
Note that (\ref{CurvBianchi}) follows using (\ref{notorsion}) and
the fact that $\lcf \cdot, \cdot \rcf_{E}$ satisfies the Jacobi
identity.

\begin{remark} \label{RemBracket2}
{\rm A situation which will be useful in the examples is the case
when we start with a vector bundle $\tau_E: E \rightarrow Q$ with a
skew-symmetric algebroid structure $(B_E, \rho_E)$ such that the
anchor map $\rho_{E}: E \rightarrow TQ$ is a skew-symmetric
algebroid morphism, that is,
\[ \rho_E \left( \lcf \sigma, \sigma' \rcf_E
\right) = [\rho_E(\sigma) , \rho_E(\sigma')]\; .
\]
Observe that this condition does not imply that $E\to Q$ is a Lie algebroid as in the previous remark.
 Under this weaker condition it is still
possible to choose the tensor $R$ in such a way  the bracket defined
in equation (\ref{BracketEtE*}) is  again the usual bracket defined
in (\ref{UsualBracket}). }\end{remark}

\begin{proof}

[Theorem \ref{TeoHamiltonian}]

Suppose that $(q^i)$ are local coordinates on $Q$,
$\{\sigma_{\a} \}$ is a local basis of $\Gamma(\tau_E)$ and that
$(B_E)^{\gamma}_{\a \beta}$, $(\rho^l_E)^i _{\a}$ and $(\rho^r_E)^i
_{\a}$ are the local structure functions of $E$ with respect to the
local coordinates $(q^i)$ and to the basis $\{\sigma_{\a} \}$. Then, from Proposition \ref{PropBracket},
it is clear that
\begin{eqnarray*} (B_E)^{\gamma}_{\a \beta} = {(D^l)}^{\gamma}_{\a
\beta} - {(D^r)}^{\gamma}_{\beta \a} \label{BGamma}
\end{eqnarray*} with $$D^l_{\sigma_{\a}} \sigma_{\beta} = {(D^l)}^{\mu}_{\a
\beta} \sigma_{\mu}, \ \ \ \ \ \ \ \ \ \ D^r_{\sigma_{\a}}
\sigma_{\beta} = {(D^r)}^{\mu}_{\a \beta} \sigma_{\mu}. $$

Moreover, if $(q^i, p_{\a})$ are the corresponding local
coordinates on $E^*$, we have that (see (\ref{SigmaalpDh}) and
(\ref{A.39.2}))
\begin{equation}\label{LocExp}
(\sigma_{\a})^{(D^l){\mathbf h}} =
(\rho^l_E)^i_{\a}\frac{\partial}{\partial q^i} +
{(D^l)}^{\gamma}_{\a \beta} p_{\gamma} \frac{\partial}{\partial
p_{\beta}},  \ \ \ \ \ \ \ \ (\sigma_{\a})^{(D^r){\mathbf h}} =
(\rho^r_E)^i_{\a}\frac{\partial}{\partial q^i} +
{(D^r)}^{\gamma}_{\a \beta} p_{\gamma} \frac{\partial}{\partial
p_{\beta}}, \end{equation}
$$(\sigma^{\a})^{\mathbf v} = \frac{\partial}{\partial p_{\a}}.$$

Now, let $H: E^*\longrightarrow \R$ be a hamiltonian function. From
(\ref{23''}), (\ref{Anclas}) and (\ref{LocExp}), it follows that
$$d^r_{\TT_l^EE^*} H = \left( \frac{\partial H}{\partial q^i}(\rho^r_E)^i_{\a}
+ \frac{\partial H}{\partial p_{\beta}} {(D^r)}^{\gamma}_{\a
\beta} p_{\gamma} \right) \left({(\sigma_{\a})}_l^{(D^l){\mathbf
h}} \right) ^* + \frac{\partial H}{\partial p_{\alpha}}
\left({(\sigma^{\a})}^{\mathbf v}_l \right) ^*.$$  Therefore, from
(\ref{24'}) and (\ref{Explocalsym}), we obtain that the right
Hamiltonian section $\HH^{( \Omega_{\TT_l^EE^*},r)}_H$ of $H$ with
respect to $\Omega_{\TT_{l}^EE^*}$ is
$$\HH^{(
\Omega_{\TT_l^EE^*},r)}_H = \frac{\partial H}{\partial p_{\a}}
(\sigma_{\a})^{(D^l){\mathbf h}}_l - \left( \frac{\partial
H}{\partial q^i}(\rho^r_E)^i_{\a} + \frac{\partial H}{\partial
p_{\beta}} {(D^r)}^{\gamma}_{\a \beta} p_{\gamma} \right)
{(\sigma^{\a})}^{\mathbf v}_l.$$

Using (\ref{Anclas}) and (\ref{LocExp}), the right Hamiltonian section yields
the left-right Hamiltonian vector field of $H$ which is
\begin{eqnarray}
\rho^l_{\TT_l^EE^*} ( \HH^{( \Omega_{\TT_l^EE^*},r)}_H) &=&
\frac{\partial H}{\partial p_{\a}} {(\rho_E^l)}_{\a}^i
\frac{\partial}{\partial q^i} - \left( \frac{\partial H}{\partial
q^i}(\rho^r_E)^i_{\beta} - \frac{\partial H}{\partial p_{\a}}
({(D^l)}^{\gamma}_{\a \beta} - {(D^r)}^{\gamma}_{\beta \a}
)p_{\gamma} \right)
\frac{\partial}{\partial p_{\beta}}\, , \nonumber \\
&=& \frac{\partial H}{\partial p_{\a}} {(\rho_E^l)}_{\a}^i
\frac{\partial}{\partial q^i} - \left( \frac{\partial H}{\partial
q^i}(\rho^r_E)^i_{\beta} - \frac{\partial H}{\partial p_{\a}}
(B_E)^{\gamma}_{\a \beta} p_{\gamma} \right)
\frac{\partial}{\partial p_{\beta}}\, \label{LRhamiltonian}.
\end{eqnarray}

Consequently, from (\ref{LocHvf}) and (\ref{LRhamiltonian}), we deduce that
$$\HH^{( \Omega_{\TT_l^EE^*},lr)}_H = \rho^l_{\TT_l^EE^*} (\HH^{(
\Omega_{\TT_l^EE^*},r)}_H) = \HH^{\Pi_{E^*}}_H.$$
\end{proof}

\begin{remark}{\rm Theorem \ref{TeoHamiltonian} was proved by Popescu
{\it et al.} \cite{PoPo} for the particular case when $E$ is a
 skew-symmetric algebroid and by de Le\'on {\it et al.} \cite{LeMaMa} for
 the particular case when $E$ is a Lie algebroid.
} \end{remark}

\bigskip

\section{Closedness of the exact symplectic section}\label{Section5}
\setcounter{equation}{0}

In order to analyze the closedness of the exact symplectic section $\Omega_{\TT_{l}^EE^*}$ we have to define a differential over tensors of type (0,2) on the algebroid $\TT_{l}^EE^*$. It is clear that it is not possible to induce a direct extension of the differential defined in (\ref{Difflr}). Therefore, the idea is to define a differential over skew-symmetric tensors and another differential over symmetric tensors of type (0,2).

\

Consider an algebroid structure $(B_{\bar E}, \rho^l_{\bar E},
\rho^r_{\bar E})$ over the vector bundle $\tau_{\bar E}: {\bar E}
\rightarrow \bar Q$. Then, it is induced over the vector bundle
$\tau_{\bar E}:{\bar E}\rightarrow \bar Q$ a skew-symmetric algebroid
$(B_{\bar E}^A , \rho^A_{\bar E})$ given by
$$B_{\bar E}^A (\sigma, \bar \sigma) = \frac{1}{2} \left(B_{\bar E}(\sigma, \bar \sigma) -
B_{\bar E}(\bar \sigma, \sigma) \right), \qquad \ \rho_{\bar E}^A
(\sigma) = \frac{1}{2} \left( \rho_{\bar E}^l(\sigma) +
\rho^r_{\bar E}(\sigma) \right), $$ and also a symmetric algebroid
$(B_{\bar E}^S , \rho^S_{\bar E})$
$$B_{\bar E}^S (\sigma, \bar \sigma) = \frac{1}{2} \left(B_{\bar E}(\sigma, \bar \sigma) + B_{\bar E}(\bar \sigma, \sigma) \right) \qquad \ \rho_{\bar E}^S (\sigma) = \frac{1}{2} \left( \rho_{\bar E}^l(\sigma) - \rho^r_{\bar E}(\sigma) \right)$$for $\sigma, \bar \sigma \in \Gamma(\tau_{\bar E}).$

Then, on a skew-symmetric tensor $T^A \in
\Gamma(\tau_{\bigwedge^2\bar E^*})$ the {\it skew-symmetric
differential} $d_{\bar E} ^A$ is defined by
\begin{eqnarray*} (d_{\bar E}^A \, T^A) (\sigma, \bar \sigma, \bar{\bar \sigma}) &=& \rho_{\bar E}^A(\sigma)(T^A (\bar \sigma, \bar{\bar \sigma})) - \rho_{\bar E}^A(\bar \sigma)(T^A (\sigma, \bar{\bar \sigma})) + \rho_{\bar E}^A(\bar{\bar \sigma})(T^A (\sigma, \bar \sigma))\\ && - T^A(B_{\bar E}^A(\sigma, \bar \sigma), \bar{\bar \sigma}) + T^A(B_{\bar E}^A(\sigma, \bar{\bar \sigma}), \bar \sigma) - T^A(B_{\bar E}^A(\bar \sigma, \bar{\bar \sigma}), \sigma)\end{eqnarray*}
and on a symmetric tensor $T^S$ in $\Gamma(\tau_{\otimes_2^0\bar
E^*})$ the {\it symmetric differential} $d_{\bar E}^S$ is
\begin{eqnarray*} (d_{\bar E}^S \, T^S) (\sigma, \bar \sigma, \bar{\bar \sigma}) &=& \rho_{\bar E}^S(\sigma)(T^S (\bar \sigma, \bar{\bar \sigma})) + \rho_{\bar E}^S(\bar \sigma)(T^S (\sigma, \bar{\bar \sigma})) + \rho_{\bar E}^S(\bar{\bar \sigma})(T^S(\sigma, \bar \sigma))\\ && - T^S(B_{\bar E}^S(\sigma, \bar \sigma), \bar{\bar \sigma}) - T^S(B_{\bar E}^S(\sigma, \bar{\bar \sigma}), \bar \sigma) - T^S(B_{\bar E}^S(\bar \sigma, \bar{\bar \sigma}), \sigma),\end{eqnarray*} for $\sigma, \bar \sigma, \bar{\bar \sigma} \in \Gamma(\tau_{\bar E})$.

Note that $d_{\bar E}^A \, T^A$ (respectively, $d_{\bar E}^S\,
T^S$) is a skew-symmetric tensor of type (0,3) (respectively, a
symmetric tensor of type (0,3)).

Now, we may extend the definition of the differential to any
(0,2)-tensor in $\Gamma(\tau_{\otimes_2^0\bar E^*})$. If $T$ is a
section of the vector bundle $\tau_{\otimes_2^0\bar E^*}:
\otimes_2^0\bar E^* \rightarrow \bar Q$ then the differential of $T$ is
the section $d_{\bar E}^{AS} T$ of the vector bundle
$\tau_{\otimes_3^0\bar E^*} : \otimes_3^0\bar E^* \rightarrow \bar Q$
defined by
\begin{equation}( d_{\bar E}^{AS}\, T) (\sigma, \bar \sigma,
\bar{\bar \sigma})  = ( d_{\bar E}^A \, T^A) (\sigma, \bar \sigma,
\bar{\bar \sigma}) + ( d_{\bar E}^S\, T^S) (\sigma, \bar \sigma,
\bar{\bar \sigma}) \end{equation} where $$T^A (\sigma, \bar
\sigma) = \frac{1}{2}(T(\sigma, \bar \sigma) - T(\bar \sigma,
\sigma)) \quad \mbox{and} \quad T^S (\sigma, \bar \sigma) =
\frac{1}{2}(T(\sigma, \bar \sigma) + T(\bar \sigma, \sigma)), $$
for $\sigma, \bar \sigma, \bar{\bar \sigma} \in \Gamma(\tau_{\bar
E})$. Note that $T^A$ and $T^S$ are the skew-symmetric and symmetric part of the tensor $T$.

\begin{theo}\label{Teocloded}
Consider an algebroid structure $(B_E, \rho^l_E, \rho^r_E)$ over
the vector bundle $\tau_E: E \rightarrow Q$ and the algebroid
structure $(B_{\TT_l^EE^*}, \rho^l_{\TT_l^EE^*},
\rho^r_{\TT_l^EE^*})$ induced over the vector bundle
$\tau_{\TT_l^EE^*}: {\TT_l^EE^*} \rightarrow E^*$ as in
(\ref{Anclas}) and (\ref{BracketEtE*}). If the (1,3)-tensor $R$ in
equation (\ref{BracketEtE*}) verifies the following relations
\begin{equation}R(\sigma, \bar \sigma) \bar{\bar \sigma} = -
R(\bar \sigma, \sigma)\bar{\bar \sigma} \label{CurvSkew1}
\end{equation} and
\begin{equation} R(\sigma, \bar \sigma) \bar{\bar \sigma} + R(\bar \sigma, \bar{\bar \sigma}) \sigma + R(\bar{\bar \sigma}, \sigma) \bar \sigma =0 \qquad \hbox{(first Bianchi identity)},\label{CurvBianchi1} \end{equation}
 for all $\sigma, \bar \sigma, \bar{\bar \sigma}\in \Gamma({E})$,
 then the exact symplectic section $\Omega_{\TT_l^EE^*}$ is closed, that is $$d_{\TT_l^EE^*}^{A}\, \Omega_{\TT_l^EE^*} = 0. $$
\end{theo}

\begin{proof}
First, we are going to build the skew-symmetric algebroid induced on $\tau_{\TT_l^EE^*} : \TT_l^EE^* \rightarrow E^*$. The anchor map of this skew-symmetric algebroid is given by
\begin{equation}
\begin{array}{l}
\rho^A_{\TT_l^EE^*} (\sigma^{(D^l){\mathbf h}}_l) = \frac{1}{2} (\sigma^{(D^l){\mathbf h}} + \sigma^{(D^r){\mathbf h}})\\ \vspace{-0.3cm} \\
\rho^A_{\TT_l^EE^*}(\kappa_l^{\mathbf v}) = \kappa^{\mathbf v}.
\end{array} \end{equation}
and the skew-symmetric bracket is
\begin{equation}
\begin{array}{l}
B^A_{\TT_l^EE^*} (\sigma^{(D^l){\mathbf h}}_l , \bar
\sigma^{(D^l){\mathbf h}}_l) = B^A_E( \sigma , \bar
\sigma)^{(D^l){\mathbf h}}_l + R(\sigma, \bar \sigma,
\cdot)^{\mathbf v}_l  \\ \vspace{-0.3cm} \\
B^A_{\TT_l^EE^*} (\sigma^{(D^l){\mathbf h}}_l , \kappa_l^{\mathbf v}) = \frac{1}{2} (D^l_{\sigma}  \kappa + D^r_{\sigma}  \kappa)^{\mathbf v}_l  = -
B^A_{\TT_l^EE^*} (\kappa_l^{\mathbf v}, \sigma^{(D^l){\mathbf h}}_l)  \\ \vspace{-0.3cm} \\
B^A_{\TT_l^EE^*}(\kappa_l^{\mathbf v}, \nu_l^{\mathbf v}) = 0
\end{array} \label{BracketA}
\end{equation}  for $\sigma, \bar \sigma \in \Gamma(\tau_E)$ and $\kappa, \nu \in \Gamma(\tau_{E^*})$.

To obtain the first equation of (\ref{BracketA}) we used the fact that the tensor $R$ in (\ref{BracketEtE*}) verifies the skew-symmetric property: $R(\sigma,\bar \sigma)=-R(\bar \sigma,\sigma)$.

In order to compute the skew-symmetric differential of the tensor $\Omega_{\TT_l^EE^*}$, consider a local basis $\{\sigma_\a\}$ of $\Gamma(\tau_E)$ and the corresponding dual basis $\{\sigma^{\a}\}$ of $\Gamma(\tau_{E^*})$.
A straightforward computation gives that
\begin{eqnarray}
(d_{\TT_l^EE^*}^{A} \Omega_{\TT_l^EE^*}) \left( (\sigma_\a)^{(D^l){\mathbf h}}_l ,
(\sigma_\beta)^{(D^l){\mathbf h}}_l, (\sigma_\gamma)^{(D^l){\mathbf h}}_l \right)
& = & \widehat{R(\sigma_\a, \sigma_\beta) \sigma_\gamma} + \widehat{R(\sigma_\beta,
\sigma_\gamma) \sigma_\a} \nonumber \\ && + \widehat{R(\sigma_\gamma, \sigma_\a) \sigma_\beta} =  0, \label{dAHoriz} \\
(d_{\TT_l^EE^*}^{A}\Omega_{\TT_l^EE^*}) \left( (\sigma_\a)^{(D^l){\mathbf h}}_l ,
(\sigma_\beta)^{(D^l){\mathbf h}}_l,(\sigma^\gamma)_l^{\mathbf v}\right) & = & 0, \label{dAHorizVert} \\
(d_{\TT_l^EE^*}^{A}\Omega_{\TT_l^EE^*}) \left( (\sigma_\a)^{(D^l){\mathbf h}}_l ,
(\sigma^\beta)_l^{\mathbf v}, (\sigma^\gamma)_l^{\mathbf v} \right) &=&
0, \nonumber\\
(d_{\TT_l^EE^*}^{A}\Omega_{\TT_l^EE^*}) \left(
(\sigma^\a)_l^{\mathbf v} , (\sigma^\beta)_l^{\mathbf v},
(\sigma^\gamma)_l^{\mathbf v} \right) &=& 0, \nonumber
\end{eqnarray}
where for the proof of (\ref{dAHorizVert}) we have used the fact that the bracket $B_E^A$ can be written as $$B_E^A (\sigma_\a , \sigma_\beta) = \frac{1}{2} \left( D^l_{\sigma_\a} \sigma_\beta - D^l_{\sigma_\beta}\sigma_\a + D^r_{\sigma_\a} \sigma_\beta - D^r_{\sigma_\beta} \sigma_\a \right).$$ Since the (0,3)-tensor $d_{\TT_l^EE^*}^{A}\Omega_{\TT_l^EE^*}$ is skew-symmetric the proof is complete.

\end{proof}

Natural choices of a $(1,3)$-tensor field $R$ verifying properties
(\ref{CurvSkew1}) and  (\ref{CurvBianchi1}) are $R\equiv 0$ and
the curvature $R=R^{\nabla^{\GG}}$ of a Levi-Civita connection in
the case when $E$ is a Lie algebroid with a bundle metric $\GG$
(see (\ref{CurvSkew}) and (\ref{CurvBianchi})). From the last
case, it is possible to construct new direct examples of a tensor
field of type $(1,3)$ satisfying (\ref{CurvSkew1}) and
(\ref{CurvBianchi1}). Consider a
 Lie algebroid $\tau_E: E\to Q$ with Lie algebroid
structure $(\lcf \cdot, \cdot \rcf_{E}, \rho_{E})$ and
$R_E=R^{\nabla^{\GG}}$ the curvature of the Levi-Civita connection
associated to a bundle metric $\GG$. Take now
 a vector subbundle $\tau_D: D\to Q$ of $E$, $i_{D}: D\to E$ being the canonical inclusion, equipped with an algebroid structure
 $(B_{D}, \rho^l_{D}, \rho^r_{D})$
and an arbitrary vector bundle morphism $F: E\to D$. Then, we may
construct the $(1,3)$-tensor field:
\[
R_D(\sigma, \bar \sigma) \bar{\bar
\sigma}=F(R^{\nabla^{\GG}}(i_D\circ\sigma, i_D\circ\bar \sigma)
(i_D\circ \bar{\bar \sigma}))
\]
for all $\sigma, \bar \sigma, \bar{\bar \sigma}\in \Gamma(D)$. It
follows that $R_D$ satisfies both conditions (\ref{CurvSkew1})
and  (\ref{CurvBianchi1}). Observe, for instance,  that it is
precisely  the case of nonholonomic mechanics discussed in
Subsection \ref{qaz}, where now $F$ is the orthogonal projector
$P$.

\begin{remark} {\rm
There is a natural extension of symmetric and skew-symmetric
differentials on tensors of type $(0,k)$. That is,
 on $\Psi^A \in \Gamma(\tau_{\bigwedge^k\bar E^*})$
 the {\it skew-symmetric differential} $d_{\bar E} ^A$ of $\Psi^A$
 is a section of $\tau_{\bigwedge^{k+1}\bar E^*} : \bigwedge^{k+1}
\bar E^* \rightarrow Q$ defined by
\begin{eqnarray*} (d_{\bar E}^A \, \Psi^A) (\sigma_0,\sigma_1,..., \sigma_k) &=&
\sum_{i=1}^k (-1)^i \rho_{\bar E}^A(\sigma_i)(\Psi^A (\sigma_0,..., \widehat \sigma_i,..., \sigma_k))
\\ && + \sum_{i<j} (-1)^{i+j} \Psi^A(B_{\bar E}^A(\sigma_i, \sigma_j),
\sigma_0,\sigma_1, ..., \widehat \sigma_i,...,\widehat
\sigma_j,...,\sigma_k)\end{eqnarray*} and on a symmetric tensor
$\Psi^S$ in $\Gamma(\tau_{\otimes_k^0\bar E^*})$ the {\it
symmetric differential} $d_{\bar E}^S\Psi^S$ of $\Psi^S$ is the
symmetric tensor in $\Gamma(\tau_{\otimes_{k+1}^0\bar E^*})$
defined by
\begin{eqnarray*} (d_{\bar E}^S \, \Psi^S) (\sigma_0,\sigma_1,..., \sigma_k) &=&
\sum_{i=1}^k  \rho_{\bar E}^S(\sigma_i)(\Psi^S (\sigma_0,..., \widehat \sigma_i,..., \sigma_k))
\\ && - \sum_{i<j} \Psi^S (B_{\bar E}^S(\sigma_i, \sigma_j),
\sigma_0,\sigma_1, ..., \widehat \sigma_i,...,\widehat \sigma_j,...,\sigma_k),
\end{eqnarray*}
for $\sigma_0,\sigma_1,..., \sigma_k \in \Gamma(\tau_{\bar E})$.

Note that
\[
d_{\bar{E}}^A(\Psi^A \wedge \mu^A) = d_{\bar{E}}^A \Psi^A \wedge
\mu^A + (-1)^k \Psi^A \wedge d_{\bar{E}}^A\mu^A, \; \;
d_{\bar{E}}^S(\Psi^S \vee \mu^S) = d_{\bar{E}}^S \Psi^S \vee \mu^S
+ \Psi^S \vee d_{\bar{E}}^S\mu^S,
\]
for $\Psi^A \in \Gamma(\tau_{\Lambda^k\bar{E}^*})$, $\mu^A \in
\Gamma(\tau_{\Lambda^l\bar{E}^*})$, $\Psi^S \in
\Gamma(\tau_{\otimes^0_k\bar{E}^*})$, $\mu^S \in
\Gamma(\tau_{\otimes^0_l\bar{E}^*})$, with $\Psi^S$ and $\mu^S$
symmetric tensors and $\vee$ being the symmetric product. On the
other hand, if $\alpha, \beta \in \Gamma(\tau_{\bar{E}^*})$ we
have that
\[
\alpha \otimes \beta = \displaystyle \frac{1}{2} (\alpha \wedge
\beta + \alpha \vee \beta)
\]
and thus
\[
d^{AS}_{\bar{E}}(\alpha \otimes \beta) = \displaystyle \frac{1}{2}
(d_{\bar{E}}^A \alpha \wedge \beta - \alpha \wedge
d_{\bar{E}}^A\beta)+ \frac{1}{2} (d_{\bar{E}}^S\alpha \vee \beta +
\alpha \vee d_{\bar{E}}^S\beta).
\]

}
\end{remark}

\begin{remark}{\rm
\begin{enumerate}
\item
The {\it skew-symmetric differential} was defined in \cite{LeMaMa}
as the {\it almost differential} on an {\it almost Lie algebroid}.
Note that   $(d^A_{\bar E} )^2 = 0 $ if and only if $(B_{\bar{E}}^A,
\rho_{\bar{E}}^A)$ is a Lie algebroid structure on $\tau_{\bar
E}:\bar E \rightarrow Q$.

\item
Let $\bar{E}$ be the tangent bundle of the manifold $Q$ and
$\nabla$ be a linear connection on $Q$. Then, $(B_{TQ}^{\nabla},
id_{TQ}, -id_{TQ})$ is a symmetric algebroid structure on $TQ$,
where
\[
B_{TQ}^{\nabla}(X, Y) = \nabla_{X}Y + \nabla_{Y}X, \; \; \mbox{
for } X, Y \in {\frak X}(Q).
\]
Moreover, the corresponding symmetric differential $d_{TQ}^S$ was
considered in \cite{HeBoPe}. In fact, in \cite{HeBoPe} using the
symmetric differential and the symmetric Lie derivative, the
derivations of the algebra of symmetric tensors are classified and
the Fr\"olicher-Nijenhuis bracket for vector valued symmetric
tensors is introduced. This theory is the symmetric counterpart of
the theory of vector valued differential forms which was developed
by Fr\"olicher-Nijenhuis \cite{FroNi}.

\end{enumerate}
}
\end{remark}

\medskip

\section{Examples revisited}
\setcounter{equation}{0}

\subsection{The symmetric case: Gradient extension of dynamical systems}
(See Subsection \ref{grad-ext}).

In this case, we have a Riemannian manifold $(Q, {\mathcal G})$
and the vector bundle $\tau_{TQ}: TQ \rightarrow Q$ endowed with
the symmetric product
\[
B_{TQ}(X, Y) = \nabla^{\mathcal G}_{X}Y + \nabla^{\mathcal
G}_{Y}X, \; \; \; \mbox{ for } X, Y \in {\mathfrak X}(Q).
\]
The anchor maps are $\rho_{TQ}^l = id_{TQ}$ and $\rho_{TQ}^r =
-id_{TQ}$. Thus,
\[
B_{TQ}(X, Y) = D^{l}_{X}Y - D^{r}_{X}Y
\]
where $D^l$ (respectively, $D^r$) is the $\rho_{TQ}^l$-connection
(respectively, the $\rho_{TQ}^r$-connection) defined by
\[
D^l_XY = \nabla^{\mathcal G}_{X}Y
\]
(respectively, $D^r_XY = -\nabla_{X}^{\mathcal G}Y$). Moreover, it
is easy to prove that the $TQ$-tangent bundle to $T^*Q$,
${\mathcal T}_{l}^{TQ}T^*Q$, may be identified with the vector
bundle $\tau_{T(T^*Q)}: T(T^*Q) \rightarrow T^*Q$. Under this
identification, we have that (see (\ref{SigmaalpDh}))
\begin{equation}\label{34'}
\left( \frac{\partial}{\partial q^i}\right)^{D^l{\mathbf h}}_{l} =
\frac{\partial}{\partial q^i} + \Gamma_{ij}^k p_{k}
\frac{\partial}{\partial p_{j}}, \quad \left(
\frac{\partial}{\partial q^i}\right)^{D^r{\mathbf h}}_{l} =
-\frac{\partial}{\partial q^i} - \Gamma_{ij}^k p_{k}
\frac{\partial}{\partial p_{j}}, \quad (dq^i)^{\mathbf v}_{l} =
\frac{\partial}{\partial p_{i}},
\end{equation}
where $(q^i, p_i)$ are fibred coordinates on $T^*Q$ and
$\Gamma_{ij}^k$ are the Christoffel symbols of the Levi-Civita
connection $\nabla^{\mathcal G}$. Therefore, using
(\ref{Explocalsym}), we deduce that the exact symplectic structure
$\Omega_{T(T^*Q)}$ is just the canonical symplectic structure of
$T^*Q$
\begin{equation}\label{34''}
\Omega_{T(T^*Q)} = dq^i \wedge dp_{i}
\end{equation}
(note that $\Gamma_{ij}^k = \Gamma_{ji}^k$).

On the other hand, from (\ref{Anclas}) and (\ref{34'}), it follows
that
\[
 \rho^r_{T(T^*Q)}\left(\frac{\partial}{\partial q^i}\right) = - \frac{\partial}{\partial
q^i} -2 \Gamma_{ij}^k p_k \frac{\partial}{\partial p_j}, \; \; \;
\rho^r_{T(T^*Q)}\left( \frac{\partial}{\partial p_i}\right) =
\frac{\partial}{\partial p_i}.
\]
Consequently, if $H \in C^{\infty}(T^*Q)$ we obtain that
\[d^r_{T(T^*Q)} H = \left( - \frac{\partial H}{\partial q^i} -2
\Gamma_{ij}^k p_k \frac{\partial H}{\partial p_j} \right) dq^i +
\frac{\partial H}{\partial p_j} \, dp_j
\]
which implies that the right-Hamiltonian section of $H$ is the
vector field on $T^*Q$ given by
\[
{\mathcal H}_{H}^{(\Omega_{T(T^*Q)}, r)} = \displaystyle
\frac{\partial H}{\partial p_{i}} \frac{\partial}{\partial q^i} +
\left(\frac{\partial H}{\partial q^i} + 2 \Gamma_{ij}^{k} p_{k}
\frac{\partial H}{\partial p_{i}}\right)\frac{\partial}{\partial
p_{i}}.
\]
Thus, if we apply the above construction to the Hamiltonian
function $H = H_{X} = \hat{X}$, with $X \in {\mathfrak X}(Q)$, we
reobtain the Hamilton equations (\ref{HamilGrad}).

\subsection{Skew-symmetric mechanics: Nonholonomic systems}
 (See Subsection \ref{qaz}).

Consider $(\mathfrak{g}, [\cdot , \cdot ]_{\mathfrak{g}})$ a Lie
algebra of finite dimension. In this case, the Lie bracket is the
Lie algebra structure $[\cdot , \cdot ]_{\mathfrak{g}}$ and the
anchor map is the null map.

Consider now a  nonholonomic mechanical system on $\mathfrak{g}$,
that is a vector subspace $\mathfrak{d} \subset \mathfrak{g}$ of
kinematic constraints ($\mathfrak{d}$ is not, in general, a Lie
subalgebra) and a lagrangian function $L:\mathfrak{g} \rightarrow
\R$ of mechanical type induced by a scalar product ${\mathcal G}$
on $\mathfrak{g}$. As we did in Example 2.2 we assert that
$(\mathfrak{d} , [\cdot, \cdot]_{\mathfrak{d}},0)$ is a
skew-symmetric algebroid with the bracket given by $[\xi,
\eta]_{\mathfrak{d}} = P([ i_{\mathfrak{d}} (\xi) ,
i_{\mathfrak{d}}(\eta) ]_{\mathfrak{g}})$ (see Subsection
\ref{qaz}).

In what follows we are going to use the formalism in $\TT^{\dd}
\dd^*$  proposed in Section \ref{section-main} to find an exact
symplectic form and the corresponding Hamilton equations.

Let us consider a basis $\{\xi_a \}$ of $\dd$ and $\{\xi^a\}$ the
dual basis of $\dd^*$.

In this case,  we choose the 0-connection $D= D^l = D^r$ to be
$D_{\xi_a} \xi_b = \frac{1}{2} [\xi_a , \xi_b ]_{\dd}$ and thus
$\Gamma_{ab}^c = \frac{1}{2} c_{ab}^c$ where $c_{ab}^c$ are the
structure constants of the skew-symmetric algebroid $(\mathfrak{d}
, [\cdot, \cdot]_{\mathfrak{d}},0)$.

Now, it is easy to prove that $\TT^{\dd}\dd^*$ may be identified
with $\dd^* \times \dd \times \dd^*$ and, under this
identification, the vector bundle projection
$$\tau_{\TT^{\dd}\dd^*} : \TT^{\dd}\dd^* \rightarrow \dd^*$$ is
just the canonical projection on the first factor
\[
pr_{1}: \dd^* \times \dd \times \dd^* \rightarrow \dd^*.
\]

Since $\dd$ satisfies the hypotheses of Remark \ref{RemBracket2},
it is easy to see that a suitable structure of skew-symmetric
algebroid on  $\TT^{\dd}\dd^* \simeq \dd^* \times \dd \times \dd^*
\to \dd^*$ is determined by the following relations:
$$B_{\TT^{\dd}\dd^*} \left( (\cdot, \sigma,
\upsilon), (\cdot, \sigma', \upsilon') \right)(\kappa) = (\kappa,
[\sigma, \sigma']_{\dd}, 0)$$ for $\kappa \in \dd^*$, $(\sigma,
\upsilon), (\sigma', \upsilon') \in \dd \times \dd^*$ and
$$\rho_{\TT^{\dd}\dd^*}^l (\kappa, \sigma, \upsilon) = (\kappa,
\upsilon).$$
A straightforward computation shows that
$$\Omega_{\TT^{\dd}\dd^*} \left( (\kappa, \sigma, \upsilon),
(\kappa, \sigma', \upsilon') \right) =  \upsilon'(\sigma) -
\upsilon(\sigma') - \kappa\left( [\sigma, \sigma']_{\dd}\right)$$
for $(\kappa, \sigma, \upsilon), (\kappa, \sigma', \upsilon') \in
\dd^* \times \dd \times \dd^*$.

Now, we consider the basis $\{E_{a}, \tilde{E}^a\}$ of
$\Gamma(\tau_{\TT^{\dd}\dd^*})$ defined as
\begin{eqnarray*} E_a &=& (\cdot
, \xi_a, 0) \ \ \ \mbox{such that} \ \ \ E_a(\kappa) = (\kappa, \xi_a,0)\\
\tilde E^a &=& (\cdot , 0, \xi^a) \ \ \ \mbox{such that} \ \ \
\tilde E^a(\kappa) = (\kappa , 0, \xi^a).
\end{eqnarray*}
Then, we obtain
$$B_{\TT^{\dd}\dd^*} (E_a, E_b) = c_{ab}^c E_c$$
$$B_{\TT^{\dd}\dd^*} (E_a, \tilde E^b) = B_{\TT^{\dd}\dd^*} (\tilde E^a, E_b)
= B_{\TT^{\dd}\dd^*} (\tilde E^a, \tilde E^b) = 0$$ and the anchor
map is $$\rho_{\TT^{\dd}\dd^*} (E_a) = 0 \ \mbox{    and    } \
\rho_{\TT^{\dd}\dd^*} (\tilde E^b) = \xi^b.$$ Thus we conclude
that
$$\Omega_{\TT^{\dd}\dd^*} = -\frac{1}{2} c_{ab}^c p_c E^a \wedge E^b
+ E^a \wedge \tilde E_a$$ where $\{E^a, \tilde E_b\}$ is the dual
basis induced by $\{E_a, \tilde E^b\}$ and $p_a$ are the
coordinates in $\dd^*$ induced by $\xi^a$. Moreover,
$$d_{\TT^{\dd}\dd^*}H(E_a) = 0 \ \mbox{    and    } \ d_{\TT^{\dd}\dd^*}H(\tilde E^a) =
\frac{\partial H}{\partial p_a}.$$

Therefore, we have that the unique solution of Equation
\begin{equation}\label{symplectic-equation-Lie}
i_X \Omega_{\TT^{\dd}\dd^*}=d_{\TT^{\dd}\dd^*} H\, ,
\end{equation}
is
\[
\HH^{\Omega_{\TT^{\dd}\dd^*}}_H=\frac{\partial H}{\partial p_a}E_a+c_{ab}^cp_c\frac{\partial H}{\partial p_b}\tilde{E}^a
\]
Now, using the anchor map we obtain that the corresponding
Hamiltonian vector field on $\dd^*$:
 $$\HH^{\Pi_{\dd^*}}_H =c_{ab}^cp_c\frac{\partial H}{\partial p_b}\xi^a= \rho_{\TT^{\dd}\dd^*} (\HH^{
\Omega_{\TT^{\dd}\dd^*}}_H).
$$

Its integral curves are precisely the \emph{nonholonomic
Lie-Poisson equations} (see \cite{CdLMM2007} and references
therein)
$$\dot p_a = c_{ab}^c p_c \frac{\partial H}{\partial p_b},$$ that is,
using a classical notation,
$$\dot \kappa = ad^{\dd^*}_{\ \frac{\partial H}{\partial \kappa}}
\kappa,
\; \; \mbox{ for } \kappa \in \dd^*
$$
where $ad^{\dd^*} : \dd \times \dd^* \rightarrow
\dd^*$ is the map defined as $(ad^{\dd^*}_{\xi} (\kappa)) (\eta) =
\kappa([\xi, \eta]_{\dd})$ for $\xi, \eta \in \dd$ and $\kappa \in
\dd^*$. Note that if $\dd = \mathfrak{g}$ then $ad^{\dd^*} = ad^*$
is the infinitesimal coadjoint representation.

\subsection{Mixed mechanics: }
\ \

\noindent 6.3.1 {\bf Generalized nonholonomic mechanics on a Lie
algebra}. (See Subsection 2.3.1).

As in the previous example, consider a Lie algebra $(\mathfrak{g},
[\cdot , \cdot ]_{\mathfrak{g}})$ of finite dimension, a subspace
$\mathfrak{d} \subset \mathfrak{g}$ and a lagrangian $L:\mathfrak
g \rightarrow \R$ of mechanical type induced by a scalar product
${\mathcal G}$ on $\mathfrak g$. Since we are considering a
generalized nonholonomic system, $\dd$ is endowed with an
algebroid structure $(\dd, B_{\dd}, 0,0)$ given by
(\ref{GNHbracket}), (\ref{GNHanchorl}) and (\ref{GNHanchorr}), (in
this case, the anchors are zero but the bracket is not necessarily
skew-symmetric). In fact,
$$B_{\dd}(\sigma, \sigma') = P[ \sigma,
\Pi(\sigma')]_{\mathfrak g}$$ for $\sigma, \sigma' \in \dd$ and $P
: \mathfrak g = \dd \oplus \dd^{\perp} \rightarrow \dd \subseteq
\mathfrak g$, $\Pi: \mathfrak g = \tilde \dd \oplus \dd^{\perp}
\rightarrow \tilde \dd \subseteq \mathfrak g$ the corresponding
projectors.

As in the previous example, the space $\TT^{\dd}\dd^*$ may be
identifiaed with the product $\dd^* \times \dd \times \dd^*$ and,
under this identification, the vector bundle projection is the
canonical projection on the first factor
$$pr_1 :\dd^* \times \dd \times \dd^*
\rightarrow \dd^*.$$ Then, we obtain that a suitable bracket on
$\Gamma(\TT^{\dd}\dd^*)$ has the following form
$$B_{\TT^{\dd}\dd^*} \left( (\cdot, \sigma, \upsilon), (\cdot,
\sigma', \upsilon') \right)(\kappa) = (\kappa, B_{\dd}(\sigma,
\sigma'), R(\kappa, \sigma, \upsilon, \sigma', \upsilon')),$$ for
$\kappa \in \dd^*$, $(\sigma, \upsilon), (\sigma', \upsilon') \in
\dd^*$, with $R(\kappa, \sigma, \upsilon, \sigma', \upsilon') \in
\dd^*$. The anchor maps, in this case, are
\begin{eqnarray*} \rho_{\TT^{\dd}\dd^*}^l (\kappa, \sigma, \upsilon)
&=&(\kappa, \upsilon)\\
\rho_{\TT^{\dd}\dd^*}^r (\kappa, \sigma, \upsilon) &=&
\left(\kappa, \upsilon + i_{\dd}^*(\Pi^*({\nabla}^{\GG}_{\sigma}
P^*\kappa ) - {\nabla}^{\GG}_{\Pi(\sigma)} P^*\kappa) \right)
\end{eqnarray*}
where $i_{\dd}: \dd \to \mathfrak g$ is the canonical inclusion
and $\nabla^{\mathcal G}$ is the Levi-Civita connection of the
scalar product ${\mathcal G}$ on $\mathfrak g$. Thus,
$$\Omega_{\TT^{\dd}\dd^*} \left( (\kappa, \sigma, \upsilon),
(\kappa, \sigma', \upsilon') \right) = \upsilon'(\sigma) -
\upsilon(\sigma') - \kappa\left( P({\nabla}^{\GG}_{\sigma'}
\Pi(\sigma) - {\nabla}^{\GG}_{\sigma}\Pi(\sigma') )\right)$$ for
$(\kappa, \sigma, \upsilon), (\kappa, \sigma', \upsilon') \in
\dd^* \times \dd \times \dd^*$.

Considering the same basis as in the previous example $\{E_a, \tilde
E^b\}$ (but with $\{\xi_a, \xi_A\}$ an adapted basis to $\dd \oplus
\tilde \dd^{\perp}$)  and the dual basis $\{E^a, \tilde E_b\}$ we
get that the left and right anchor maps are, respectively
$$\rho^l_{\TT^{\dd}\dd^*} (E_a) = 0 \ \mbox{    and    } \
\rho^l_{\TT^{\dd}\dd^*} (\tilde E^b) = \xi^b,$$
$$\rho^r_{\TT^{\dd}\dd^*} (E_a) = -\Xi_{ab}^c p_c \xi^b  \ \mbox{    and    } \
\rho^r_{\TT^{\dd}\dd^*} (\tilde E^b) = \xi^b,$$ where $\Xi_{ab}^c
= (D^l)_{ab}^c - (D^r)_{ab}^c$, with $(D^l)_{ab}^c$ and $
(D^r)_{ab}^c$ the Christoffel symbols of the left and right
connections for the generalized nonholonomic systems (see Appendix
B).

A similar computation, as in the nonholonomic case, shows that
$$\Omega_{\TT^{\dd}\dd^*} = \frac{1}{2}(- \tilde c_{ab}^c p_c +
\Xi_{ba}^c ) E^a \wedge E^b + E^a \wedge \tilde E_a.$$ Finally by
means of the left-right Hamiltonian vector field of $H$, we get
$$\dot p_a = \tilde c_{ab}^c p_c \frac{\partial H}{\partial p_b}.$$
Therefore, we obtain \emph{the generalized nonholonomic Lie-Poisson
equations},
$$\dot \kappa = ad^{\dd^*}_{\ \frac{\partial H}{\partial \kappa}} \kappa$$
for $\kappa \in \dd^*$ and where $ad^{\dd^*} : \dd \times \dd^*
\rightarrow \dd^*$ is the map defined as $(ad^{\dd^*}_{\xi}
(\kappa)) (\eta) = \kappa(B_{\dd}(\xi, \eta))$, for $\xi, \eta \in
\dd$ and $\kappa\in \dd^*$.

\ \\

\noindent 6.3.2 {\bf Lagrangian mechanics for modifications of the
standard Lie bracket.}

Let us reconsider the Example in Subsection 2.3.2,  in the case
when the tensor field $T$ is given by $T(X, Y) = S(X, Y) - S(Y,
X)$, for $X, Y \in \mathfrak X(Q)$, where $S$ is the contorsion
tensor field induced by an affine connection $\nabla$. The
horizontal and vertical lifts induced by the connection $\nabla$
give rise to an almost Lie algebroid structure on $TT^*Q\to T^*Q$.
Straightforward computations permit to deduce that
\[
\Omega_{TT^*Q}=dq^i\wedge (dp_i-S_{ji}^kp_k dq^j)\; .
\]
Since
\begin{eqnarray*}
d^r_{TT^*Q}H&=&\frac{\partial H}{\partial q^i}\,
dq^i+\frac{\partial H}{\partial p_i}\, dp_i,
\end{eqnarray*}
then the hamiltonian vector field
\begin{equation}\label{almostsymplectic}
i_{{\mathcal
H}_H^{(\Omega_{T(T^*Q)},lr)}}\Omega_{TT^*Q}=d^r_{TT^*Q}H
\end{equation}
is
\[
{\mathcal H}_{H}^{(\Omega_{T(T^*Q)}, lr)} =\frac{\partial
H}{\partial p_i} \frac{\partial }{\partial q^i}-\left(\frac{\partial
H}{\partial q^i}-(S_{ji}^k-S_{ij}^k)p_k\frac{\partial H}{\partial
p_k}\right) \frac{\partial }{\partial p_i}.
\]
Thus, the integral curves of ${\mathcal H}_{H}^{(\Omega_{T(T^*Q)},
lr)}$ are just the solutions of Eqs.(\ref{H-modif}). Observe that
Equation (\ref{almostsymplectic}) exactly reproduces the almost
symplectic realization of generalized Chaplygin systems (see
\cite{Cort}).

\bigskip

\section{Conclusions and future work}

A symplectic realization of the Hamiltonian dynamics on an
algebroid is derived. In fact, we prove that Hamiltonian systems
on an algebroid can be described by a symplectic equation
constructed in the same way than in the standard one. For this
purpose, the theory of generalized connections on an anchored
vector bundle $\tau_{E}: E \rightarrow Q$ is widely used. In
particular, we used the corresponding theory of horizontal and
vertical lifts of tensor fields on $\tau_{E}: E \rightarrow Q$ to
vector fields on the dual vector bundle $E^*$. Taking into account
that there exists a lot of examples of Hamiltonian systems on an
algebroid (gradient systems, nonholonomic mechanical systems,
generalized nonholonomic mechanical systems,...), the above
results show the ubiquity of the symplectic Hamiltonian equations
in Mechanics.

In this paper, we suppose that the constraints (kinematic or
variational) are linear. It would be interesting to discuss the
more general case when the constraints are not linear and, more
precisely, the case of affine constraints.

Another goal we have proposed is to develop a Klein formalism for
Lagrangian systems on algebroids.

Finally, a different aspect on which we intend to work is a
Hamilton-Jacobi theory for Hamiltonian systems on algebroids.

\bigskip

\appendix

\section*{Appendix A: anchored vector bundles, connections and horizontal (vertical) lifts}
\setcounter{equation}{0}
\newcounter{apendice}
\setcounter{apendice}{0}
\def\theequation{A.\arabic{equation}}
\renewcommand{\thesection}{A}
\label{SubSecLifts}

\begin{definition}\cite{CaLa}\label{defAnchored}
An {anchored vector bundle} is a real vector bundle $\tau_E : E
\rightarrow Q$ over a manifold $Q$ and a vector bundle morphism
$\rho _E : E\rightarrow TQ$. The map $\rho_E : E \rightarrow TQ$
is called the {anchor map} of the anchored vector bundle.
\end{definition}

Now, suppose that $(\tau_E: E \rightarrow Q, \rho_E)$ is an anchored  vector bundle over $Q$ and denote by $\Gamma(\tau_E)$ the space of
$C^{\infty}$-sections of the vector bundle $\tau_E: E \rightarrow Q$.

\begin{definition} \cite{CaLa} \label{defRhoconn}
A {$\rho_E$-connection} on the anchored vector bundle $(\tau_E: E
\rightarrow Q, \rho_E)$
 is a $\R$-bilinear map $D :\Gamma(\tau_E) \times \Gamma(\tau_E) \rightarrow \Gamma(\tau_E)$ such that
\begin{equation} D_{f \sigma} \gamma = f D_{\sigma} \gamma \ \ \ \ \mbox{and} \ \ \ \
D_{\sigma}(g \gamma) = \rho_E(\sigma)(g) \gamma + g D_{\sigma} \gamma \label{Condcon}
\end{equation} for $f \in C^{\infty}(Q)$, and $\sigma, \gamma \in \Gamma(\tau_E)$.
\end{definition}

\begin{remark} \label{remRhoconn}
{\rm Every vector bundle $\tau_E : E \rightarrow Q$ admits a
$\rho_E$-connection. In fact, let $\nabla : {\mathfrak X}(Q)
\times \Gamma(\tau_E) \rightarrow \Gamma(\tau_E)$ be an standard
linear connection on $\tau_E : E \rightarrow Q$. Then, if we
define the map $D: \Gamma(\tau_E) \times \Gamma(\tau_E)
\rightarrow \Gamma(\tau_E)$ as $$ D_{\sigma} \gamma =
\nabla_{\rho_E(\sigma)} \gamma \ \ \ \ \mbox{for} \ \sigma, \gamma
\in \Gamma(\tau_E),$$ it is easy to prove that $D$ is a
$\rho_E$-connection.}

\end{remark}

Let $D$ be a $\rho_{E}$-connection on the anchored vector bundle
$(\tau_{E}: E \rightarrow Q, \rho_{E})$. If $(q^i)$ are local
coordinates on $Q$ and $\{\sigma_{\alpha} \}$ is a local basis of
$\Gamma(\tau_E)$ then
$$ D_{f^{\a} \sigma_{\a}}(g^{\beta} \sigma_{\beta})  = \left(
f^{\a}g^{\beta} D_{\a {\beta}}^{\gamma} + f ^{\a}(\rho_E)^i_{\a}
\frac{\partial g^{\gamma}}{\partial q^i} \right) \sigma_{\gamma}
$$ for $f^{\a}, g^{\beta} \in C^{\infty}(Q)$, where
$$\rho_E (\sigma_{\a}) = (\rho_E)_{\a}^i \frac{\partial}{\partial q^i} \ \ \ \mbox{and}
\ \ \ D_{\sigma_{\a}}\sigma_{\beta} = D_{\a \beta}^{\gamma}\sigma
_{\gamma}.$$

$D_{\alpha\beta}^{\gamma}$ are the Christoffel symbols of the
connection $D$ with respect to the local basis
$\{\sigma_{\alpha}\}$.

Now, suppose that $\sigma_q$ is an element of the fiber $E_q$,
with $q \in Q$. Then, we may introduce the $\R$-linear map
$D_{\sigma_q} : \Gamma(\tau_E) \rightarrow E_q$ given by
$$D_{\sigma_q} \gamma = (D_{\sigma} \gamma)(q),\ \ \ \mbox{for} \ \gamma
\in \Gamma(\tau_E),$$ where $\sigma \in \Gamma(\tau_E)$ and
$\sigma(q) = \sigma_q$. Note that, using (\ref{Condcon}), one
deduces that the map $D_{\sigma_q}$ is well defined. Thus, if
$\kappa_q \in E^*_q$, we may consider the linear map
$$^{D{\mathbf h}}_{\kappa_q} : E_q \rightarrow T_{\kappa_q}E^*, \ \ \ \sigma_q \mapsto
{(\sigma_q)}^{D{\mathbf h}}_{\kappa_q}$$ where ${(\sigma_q)}^{D{\mathbf h}}_{\kappa_q}$ is
the tangent vector to $E^*$ at $\kappa_q$ which is characterized by
the following conditions
\begin{equation}{(\sigma_q)}^{D{\mathbf h}}_{\kappa_q} (f \circ \tau_{E^*}) = \rho_E(\sigma_q)(f)
 \ \ \mbox{and} \ \ {(\sigma_q)}^{D{\mathbf h}}_{\kappa_q}(\widehat \gamma) = \kappa_q(D_{\sigma_q}
 \gamma), \label{DefDh} \end{equation}
for $f \in C^{\infty}(Q)$ and $\gamma \in \Gamma(\tau_E)$. Here,
$\widehat \gamma : E^* \rightarrow \R$ is the linear function on $E^*$
induced by the section $\gamma$.

In particular, if $\sigma \in \Gamma(\tau_E)$ we may define the
\emph{$D$-horizontal lift to $E^*$} as the vector field
$\sigma^{D{\mathbf h}}$ on $E^*$ given by
$$\sigma^{D{\mathbf h}}(\kappa_q) = (\sigma(q))^{D{\mathbf h}}_{\kappa_q}, \ \ \ \mbox{for} \ \kappa_q \in E^*_q, \ \mbox{with} \ q \in Q.$$
It is clear that
\begin{equation}\label{A.38'}(\sigma + \sigma
')^{D{\mathbf h}} = \sigma^{D{\mathbf h}} + (\sigma')^{D{\mathbf
h}}, \ (f\sigma)^{D{\mathbf h}} = (f \circ \tau_{E^*})
\sigma^{D{\mathbf h}},
\end{equation}
 for $\sigma, \sigma' \in \Gamma(\tau_E)$ and $f \in C^{\infty}(Q)$.

Moreover, if $(q^i)$ are local coordinates on $Q$ and
$\{\sigma_{\a} \}$ is a local basis of $\Gamma(\tau_E)$, then we
have the corresponding local coordinates $(q^i, p_{\alpha})$ on
$E^*$ and
\begin{equation}\sigma_{\a}^{D{\mathbf h}} = (\rho_E)^i_{\a} \frac{\partial}{\partial q^i} +
D_{\a \beta}^{\gamma} p_{\gamma} \frac{\partial}{\partial
p_{\beta}}, \label{SigmaalpDh}
\end{equation}
(for more details, see \cite{CaLa}).

 On the other hand, if
$\kappa'_q \in E^*_q$ we may consider the \emph{standard vertical
lift}
 as the linear map
 $$
 ^{\mathbf v}_{\kappa'_q} : E^*_q \rightarrow T_{\kappa'_q}E^*, \ \ \
 \kappa_q \mapsto {(\kappa_q)}^{\mathbf v}_{\kappa'_q}
 $$
 with ${(\kappa_q)}^{\mathbf v}_{\kappa'_q}$ being
 the tangent vector to $E^*$ at $\kappa'_q$ which is characterized by the following
 conditions
 \begin{equation}\label{A.39.0}
 {(\kappa_q)}^{\mathbf v}_{\kappa'_q}(f \circ \tau_{E^*}) = 0
 \ \ \ \mbox{and} \ \ \ {(\kappa_q)}^{\mathbf v}_{\kappa'_q}(\widehat \gamma)
  = \kappa_q(\gamma(q))\end{equation}
   for $f \in C^{\infty}(Q)$ and $\gamma \in
  \Gamma(\tau_E)$.

Thus, if $\kappa \in \Gamma(\tau_{E^*}) $ is a section of the dual
vector bundle $\tau_{E^*} : E^* \rightarrow Q$ then the
\emph{vertical lift to $E^*$} is the vector field $\kappa^{\mathbf
v}$ on $E^*$ given by
$$\kappa^{\mathbf v}(\kappa'_q) = {(\kappa(q))}^{\mathbf v}_{\kappa'_q} \ \ \mbox{for} \ \kappa'_q \in E^*_q, \ \mbox{with} \ q \in Q.$$

It is clear that
\begin{equation}\label{A.39'}
(\kappa + \kappa ')^{\mathbf v} = \kappa^{\mathbf v} +
(\kappa')^{\mathbf v}, \ (f\kappa)^{\mathbf v} = (f \circ
\tau_{E^*}) \kappa^{\mathbf v},
\end{equation}
for $\kappa, \kappa' \in \Gamma(\tau_{E^*})$ and $f \in
C^{\infty}(Q)$.

Moreover, if $(q^i)$ are local coordinates on $Q$, $\{\sigma_\a
\}$ is a local basis of $\Gamma(\tau_E)$, $\{\sigma^\a\}$ is the
dual basis of $\Gamma(\tau_{E^*})$ and $(q^i, p_{\alpha})$ the
corresponding local coordinates on $E^*$ then
\begin{equation}\label{A.39.2}
{(\sigma^\a)}^{\mathbf v} = \frac{\partial}{\partial p_\a}.
\end{equation}

\begin{remark} {\rm If $\{\sigma_{\a}\}$
is a local basis of $\Gamma(\tau_E)$ and $\{\sigma^{\a}\}$ is the
dual basis of $\Gamma(\tau_{E^*})$ then $\{\sigma_{\a}^{D{\mathbf h}},
{(\sigma^{\a})}^{\mathbf v} \}$ is not, in general, a local basis of vector
fields on $E^*$. Note that $\rho_E$ is not, in general, an
epimorphism.}
\end{remark}

\begin{remark}
{\rm
The $\rho_E$-connection $D$ induces a $\rho_E$-connection $D^*$ on the
dual vector bundle $\tau_{E^*} : E^* \rightarrow Q$
 which is defined by $$(D_{\sigma}^* \kappa) (\gamma) = \rho_E(\sigma)(\kappa(\gamma))
 -\kappa(D_{\sigma}\gamma), $$ for $\sigma, \gamma \in \Gamma(\tau_E)$ and $\kappa \in \Gamma(\tau_{E^*})$.
 If $\{\sigma_\a \}$ is a local basis of
$\Gamma(\tau_E)$ and $\{\sigma^\a\}$ is the dual basis of
$\Gamma(\tau_{E^*})$ then $D^*_{\sigma_\a} \sigma^{\gamma} = - D_{\a
\beta}^{\gamma} \sigma^{\beta}$, where $D_{\a \beta}^{\gamma}$ are
the Christoffel symbols of the connection $D$.
   Therefore,
 if $\sigma \in \Gamma(\tau_E)$ it is possible to consider the corresponding $D^*$- horizontal lift to
 $E$ as a vector field $\sigma^{D^*{\mathbf h}}$ on $E$.

 }
\end{remark}

 The above results are a generalization of some lifting operations previously
 defined in \cite{YaIs,PatYano-0, PatYano} for the case $E=TQ$ and $\rho_{E} =
 \rho_{TQ} = id_{TQ}$.

\bigskip

\section*{Appendix B: generalized nonholonomic systems}
\setcounter{equation}{0}
\def\theequation{B.\arabic{equation}}

Let us consider a vector bundle $\tau_E: E \rightarrow Q$ with a
Lie algebroid structure $(\lcf\cdot ,\cdot \rcf_E, \rho_E)$. A
{\it linear generalized nonholonomic system} on $E$ is a
mechanical system determined by a regular lagrangian function $L:E
\rightarrow \R$ and two distributions, the kinematic constraints
described by a vector subbundle $\tau_D:D \rightarrow Q$ and the
variational constraints given by the vector subbundle
$\tau_{\tilde D} : \tilde D \rightarrow Q$. As we have explained
in section 2.3.1, the distribution $\tilde D$ is the subspace
where the constraint forces are doing null work. It is clear that
in the (classical) nonholonomic systems $D = \tilde D$.
Generalized nonholonomic systems were studied in
\cite{BaSo,c,CeGr,m}.

We will assume that the lagrangian is of mechanical type, that is,
we have a bundle metric ${\mathcal G}$ on $E$ and a real
$C^{\infty}$-function $V: Q \rightarrow \mathbb{R}$ such that
\[
L(e) = \displaystyle \frac{1}{2} {\mathcal G}(e, e) -
V(\tau_{E}(e)), \; \; \; \mbox{ for } e \in E.
\]
Moreover, we will assume that the following compatibility
condition holds
\[
E= D \oplus \tilde D^{\perp}
\]
where $\tilde D^{\perp}$ is the orthogonal complement of the
variational distribution $\tilde D$ with respect to the bundle
metric $\GG$.

 We have that the equations of motion of such a
system are given by $\delta L_{\gamma(t)} \in \tilde
D^0_{\tau_D(\gamma(t))}$, for a $\rho_E$-admissible curve $\gamma: I
\rightarrow D$. Then, the equations of motion are
\begin{equation} \left\{ \begin{array}{l}
\displaystyle \frac{dq}{dt} = \rho_{E} \circ \gamma, \\
\nabla^{\GG}_{{\gamma}(t)} {\gamma}(t)
 + grad_{\GG}V(q(t)) \in \tilde D^{\perp}_{q(t)}, \\
 {\gamma}(t) \in D_{q(t)},
\end{array} \right. \label{Agnhmotion}
\end{equation} where $\nabla^{\GG}$ is the Levi-Civita
connection of $\GG$, $\mbox{grad}_{\mathcal G}(V)$ is the section
of $\tau_{E}: E \rightarrow Q$ given by
\[
{\mathcal G}(\mbox{grad}_{\mathcal G}(V), \sigma) =
\rho_{E}(\sigma)(V), \; \; \; \mbox{ for } \sigma \in
\Gamma(\tau_{E}),
\]
and $q = \tau_{D} \circ \gamma$.

Suppose that $(q^i)$ are local coordinates on an open subset $U$
of $Q$ and that $\{\sigma_{\a}\} = \{\sigma_{a}, \sigma_{A}\}$ is
a basis of sections of the vector bundle $\tau_{E}^{-1}(U) \to U$
adapted to the decomposition $E=D \oplus  \tilde D^{\perp}$. We
will denote by $(q^i, v^{\a}) = (q^i, v^{a}, v^A)$ the
corresponding local coordinates on $E$. We will assume that the
bundle metric ${\mathcal G}$ can be locally written as $\GG =
\GG_{\a\beta} \sigma^{\a} \otimes \sigma^{\beta}$. We will also
assume that $\sigma_a$ (respectively, $\sigma_A$) is an
orthonormal basis of $\Gamma(\tau_{D})$ (respectively,
$\Gamma(\tau_{\tilde D ^{\perp}})$). Thus, we have that $\GG_{ab}
= \delta_{a}^{b}$ (respectively $\GG_{AB} = \delta_{A}^{B}$) and
it is easy to see that $\tilde D = span \{\sigma_d - \GG_{d A}
\sigma_{A} \}$. Then the system (\ref{Agnhmotion}) can be written
for $\gamma = v^a \sigma_a + v^A \sigma_A$ and $q = \tau_{E} \circ
\gamma$
\begin{equation} \left\{ \begin{array}{l} \GG
(\nabla^{\GG}_{\gamma(t)} \gamma(t) + grad_{\GG}V(q(t)), \sigma_d
- \GG_{d B} \sigma_{B} ) = 0 \\ \
\\ \dot{q}^i = (\rho^l_D)^i_av^a, \; \;  v^{A} = 0.
\end{array} \right. \label{Agnhmotion2}
\end{equation}

A straightforward computation shows that the system (\ref{Agnhmotion2})  is equivalent to
\begin{equation} \left\{ \begin{array}{l} \dot v^c +  v^a v^b \Gamma_{ab}^c
+ (\rho^r_D)_c^i \frac{\partial V}{\partial q^i} =0\\
\dot q^i = (\rho^l_D)_a^i v^a
\end{array} \right. \label{AgnhMotion3}
\end{equation}
where $\Gamma_{ab}^c$ are the Christoffel symbols of the Levi-Civita
connection in $\tau_E:E\rightarrow Q$ and
\begin{eqnarray}(\rho^l_D)_a^i & = & (\rho_E)_a^i \nonumber \\
(\rho^r_D)_c^i & = & \GG^{cd}((\rho_E)_d^i - \GG_{d A}
(\rho_E)_{A}^i) \label{AlocStuc}
\end{eqnarray} with $\GG^{\a\beta}$ the inverse matrix of
$\GG_{\a\beta}$ (note that ${\mathcal G}^{eC} = -{\mathcal
G}^{ef}{\mathcal G}_{fC}$ and that ${\mathcal G}^{ef} + {\mathcal
G}^{eC}{\mathcal G}_{Cf} = \delta^{e}_{f}$).

Now we can write these symbols $\Gamma_{ab}^c$ in terms of the
local structure functions of the Lie algebroid $\tau_E:E
\rightarrow Q$ using the expression
$$\Gamma_{ab}^c = \frac{1}{2}\GG^{c\a}\left([\a, a;b] + [\a,b;a] + [a,b;\a] \right)$$
where $[\a,\beta;\gamma] = \frac{\partial \GG_{\a \beta}}{\partial
q^i}(\rho_E)_{\gamma}^i + C_{\a \beta}^{\mu} \GG_{\mu \gamma}$
(see \cite{CdLMM2007,CorMar}). Then, since ${\mathcal G}^{cA} =
-{\mathcal G}^{cd}{\mathcal G}_{dA}$, it is easy to prove that
\[
\Gamma_{ab}^{c} v^a v^b = \GG^{cd} \left[ C_{db}^a + \GG_{a A}
C_{db}^A - \GG_{dA} C_{Ab}^{a} - \GG_{dA} \GG_{aB} C_{Ab}^{B} -
\GG_{dA} \frac{\partial \GG_{aA}}{\partial q^i} (\rho_E)_b^i
\right]v^a v^b.
\]
Thus, if we denote by $\tilde{C}_{bc}^a$ the real function given
by
\begin{equation}\label{Ctilde}
\tilde{C}_{bc}^a = -\GG^{cd} \left[ C_{db}^a + \GG_{a A} C_{db}^A
- \GG_{dA} C_{Ab}^{a} - \GG_{dA} \GG_{aB} C_{Ab}^{B} - \GG_{dA}
\frac{\partial \GG_{aA}}{\partial q^i} (\rho_E)_b^i \right]
\end{equation}
it follows that Eqs. (\ref{AgnhMotion3}) may be written as follows
$$ \left\{
\begin{array}{l} \dot v^c =  v^a v^b \tilde C_{bc}^a -
(\rho^r_D)_c^i \frac{\partial V}{\partial q^i} \\
\dot q^i = (\rho^l_D)_a^i v^a.
\end{array} \right.
$$
where $(\rho^l_D)_a^i$ and $(\rho^r_D)_c^i$ are defined as in
(\ref{AlocStuc}).

In what follows we are going to see how the functions $\tilde
C_{bc}^a$, $(\rho^l_D)_a^i$ and $(\rho^r_D)_c^i$ can be interpreted
as the local structure functions of an algebroid structure on
$\tau_D:D \rightarrow Q$.

First, let us consider the following projectors, $$P: D \oplus
D^{\perp} \rightarrow D \ \ \ \ \ \ \mbox{and} \ \ \ \ \ \Pi :
\tilde D \oplus D^{\perp} \rightarrow \tilde D$$ and the natural
inclusions $$i_D: D \rightarrow E \ \ \ \ \ \ \mbox{and} \ \ \ \ \
i_{\widetilde D} : \tilde D \rightarrow E.$$

\begin{prop} \label{ApLeibniz} Suppose that on the vector bundle $\tau_E:E \rightarrow
Q$ we have a Lie algebroid structure $(\lcf\cdot ,\cdot \rcf_E,
\rho_E)$. Then, on the vector subbundle $\tau_D : D \rightarrow Q$
we have an algebroid structure given by the bracket
$$B_D(\sigma, \eta) = P( \lcf i_D (\sigma ), (i_{\widetilde D} \circ \Pi)(\eta)
\rcf_E)$$ for $\sigma, \eta \in \Gamma(\tau_D)$ and the anchor maps
$$\rho^l_D = \rho_E \circ i_D \ \ \ \ \ \ \mbox{and} \ \ \ \ \ \rho^r_D =
\rho_E \circ i_{\widetilde D} \circ \Pi.$$ Moreover, in the local
basis $\{\sigma_a, \sigma_A\}$ adapted to the decomposition $D
\oplus \tilde D^{\perp}$, this algebroid $(B_D, \rho^l_D, \rho^r_D)$
has local structure functions given by (\ref{AlocStuc}) and
(\ref{Ctilde}).
\end{prop}

\begin{proof} In the local basis $\{\sigma_a, \sigma_A\}$ adapted to the decomposition $D \oplus
\tilde D^{\perp}$ we have that
\[ P(\sigma_a) = \sigma_a \ \ \ \ \
\ \mbox{and} \ \ \ \ \ P(\sigma_A) = \GG_{cA} \sigma_c.
\]
Note that $\sigma_{A} - {\mathcal G}_{cA} \sigma_{c} \in
\Gamma(\tau_{D^{\perp}})$. Moreover, since $\sigma_{a} - {\mathcal
G}^{ad}(\sigma_{d} - {\mathcal G}_{dA} \sigma_{A}) \in
\Gamma(\tau_{D^{\perp}})$, we deduce that
\[\Pi(\sigma_a) =
\GG^{ad}(\sigma_d - \GG_{dA} \sigma_A).
\]
Then it is simple to prove that $P \circ \Pi_{|D} = id_D$ and from
this it is obtained that the bracket and the anchor maps given above
define an algebroid structure on $\tau_{D}: D \rightarrow Q$.

On the other hand,
\[\begin{array}{rcl}
B_{D}(\sigma_{b}, \sigma_{c}) &=& (\displaystyle \frac{\partial
{\mathcal G}^{ca}}{\partial q^i}(\rho_{E})^i_b - \frac{\partial
{\mathcal G}^{cd}}{\partial q^i} (\rho_{E})^i_b {\mathcal
G}_{dA}{\mathcal G}_{aA} - {\mathcal G}^{cd}\frac{\partial
{\mathcal G}_{dA}}{\partial q^i}(\rho_{E})^i_b {\mathcal
G}_{cA})\sigma_{a} \\[8pt]&&+ \GG^{cd} ( C_{bd}^a + \GG_{a A} C_{bd}^A -
\GG_{dA} C_{bA}^{a} - \GG_{dA} \GG_{aB} C_{bA}^{B})\sigma_{a}.
\end{array}
\]
Now, using that ${\mathcal G}^{cd}{\mathcal G}_{dA}{\mathcal
G}_{aA} = {\mathcal G}^{ca} + \delta^c_a$, it follows that
\[
\displaystyle \frac{\partial {\mathcal G}^{ca}}{\partial q^i} =
\frac{\partial {\mathcal G}^{cd}}{\partial q^i} {\mathcal
G}_{dA}{\mathcal G}_{aA} + {\mathcal G}^{cd}\frac{\partial
{\mathcal G}_{dA}}{\partial q^i}{\mathcal G}_{aA} + {\mathcal
G}^{cd}{\mathcal G}_{dA} \frac{\partial {\mathcal
G}_{aA}}{\partial q^i}
\]
which implies that (see (\ref{Ctilde}))
\[
B_{D}(\sigma_{b}, \sigma_{c}) = \tilde{C}_{bc}^a \sigma_{a}.
\]
This ends the proof of the result.
\end{proof}

\

Now let us define the following two maps:
\begin{eqnarray}
D^l: \Gamma(\tau_D) \times \Gamma(\tau_D) \rightarrow \Gamma(\tau_D)
\ & \mbox{such that} & \ D^l_{\sigma} \bar \sigma = P
(\nabla^{\GG}_{\sigma} \Pi \bar \sigma) \label{D1} \\
D^r: \Gamma(\tau_D) \times \Gamma(\tau_D) \rightarrow
\Gamma(\tau_D) \ & \mbox{such that} & \ D^r_{\sigma} \bar \sigma =
P (\nabla^{\GG}_{\Pi \sigma} \bar \sigma). \label{D2}
\end{eqnarray}

\begin{prop} \label{iconn}
The map $D^l$ defined in (\ref{D1}) is a $\rho^l_{D}$-connection and
analogously the map $D^r$ defined in (\ref{D2}) is a
$\rho^r_{D}$-connection with $\rho^l_{D}$ and $\rho^r_{D}$ defined
as in Proposition \ref{ApLeibniz}.
\end{prop}

\begin{proof}:
It is sufficient to see that the maps $D^l$ and $D^r$ verify
equation (\ref{Condcon}). In fact, for $\sigma,\bar \sigma \in
\Gamma (\tau_D)$ and $f \in C^{\infty}(D)$, we have
\begin{eqnarray*}
D^l_{f \sigma} \bar \sigma &=& P (\nabla^{\GG}_{ f\sigma} \Pi \bar
\sigma) = f\, P (\nabla^{\GG}_{\sigma} \Pi \bar
\sigma) = f D^l_{\sigma} \bar \sigma, \\
D^l_{\sigma} f\bar \sigma &=& P (\nabla^{\GG}_{\sigma} \Pi (f \bar
\sigma) ) = P \left( f \nabla^{\GG}_{\sigma} \Pi \bar \sigma +
\rho_E(\sigma)(f) \Pi \bar \sigma \right) = \ f\, D^l_{\sigma}
\bar \sigma + \rho^l_{D} (\sigma) (f) \bar \sigma,
\end{eqnarray*}
where in the last equality we use again that $P \circ \Pi|_D =
id_D$.

On the other hand, it follows that
\begin{eqnarray*}
D^r_{f \sigma} \bar \sigma & = & P (\nabla^{\GG}_{\Pi (f\sigma)}
\bar \sigma) = P (\nabla^{\GG}_{f \Pi \sigma} \bar \sigma) = f\, P
(\nabla^{\GG}_{\Pi \sigma} \bar
\sigma) = f D^r_{\sigma} \bar \sigma, \\
D^r_{\sigma} f\bar \sigma &=& P (\nabla^{\GG}_{\Pi \sigma} f \bar
\sigma) = P \left( f \nabla^{\GG}_{\Pi \sigma} \bar \sigma +
\rho_E(\Pi \sigma)(f) \bar \sigma \right) = f\, D^r_{\sigma} \bar
\sigma + \rho^r_{D} (\sigma) (f) \bar \sigma.
\end{eqnarray*}
\end{proof}

Now, due to Proposition \ref{PropBracket} (section
\ref{section-main}), if we define
$$B_D( \sigma, \eta) = D^l_{\sigma} \eta - D^r_{\eta} \sigma,$$ for $\sigma,\eta \in \Gamma(\tau_D)$,
then $(B_D, \rho^l_{D}, \rho^r_{D})$ is an algebroid structure on
$\tau_D: D \rightarrow Q$.

Having the definitions of the left and right connections $D^l$ and
$D^r$ in terms of the Levi-Civita connection ${\nabla}^\GG$, it is
very easy to see that this bracket $B_D$ coincides with the one
defined in Proposition \ref{ApLeibniz}.

\end{document}